\documentclass[a4paper,12pt]{article}

\usepackage[margin=1in]{geometry}

\usepackage{amsmath,amssymb,amsthm}
\usepackage{graphicx}
\usepackage{dsfont}
\usepackage{hyperref}
\usepackage{color}
\usepackage{authblk}

\newcommand{\cost}{{\text{cost}}}
\newcommand{\weight}{{\text{w}}}
\newcommand{\weightbar}{{\overline{\text{w}}}}
\newcommand{\find}{\textsc{Find}}
\newcommand{\rootedfind}{\textsc{RootedFind}}
\newcommand{\set}[1]{\{#1\}}

\makeatletter
\newtheorem*{rep@theorem}{\rep@title}
\newcommand{\newreptheorem}[2]{%
	\newenvironment{rep#1}[1]{%
		\def\rep@title{#2 \ref{##1}}%
		\begin{rep@theorem}}%
		{\end{rep@theorem}}}
\makeatother

\newenvironment{lemma-repeat}[1]{\begin{trivlist}
		\item[\hspace{\labelsep}{\bf\noindent Lemma \ref{#1} }]\em }%
	{\end{trivlist}}
\newenvironment{theorem-repeat}[1]{\begin{trivlist}
		\item[\hspace{\labelsep}{\bf\noindent Theorem \ref{#1} }]\em }%
	{\end{trivlist}}

\newtheorem{theorem}{Theorem}
\newtheorem{lemma}[theorem]{Lemma}

\newtheorem{claim}[theorem]{Claim}
\newtheorem{definition}[theorem]{Definition}

\newtheorem{oq}[theorem]{Open Question}

\def\eps{\varepsilon}
\def\poly{poly}
\def\MaxD{\text{Max-Dist}}
\def\SumD{\text{Sum-Dist}}
\def\HamD{\text{Ham-Dist}}

\newcommand{\R}{\mathds{R}}
\newcommand{\cut}{\text{cost}}

\begin{document}

\title{New Hardness Results for Planar Graph Problems in P and an Algorithm for Sparsest Cut}
\author[1]{Amir Abboud}
\affil[1]{%
  IBM Almaden Research Center\\
amir.abboud@gmail.com
}
\author[2]{Vincent Cohen-Addad}
\affil[2]{%
  Sorbonne Universit\'e, UPMC Univ Paris 06, CNRS, LIP6\\
vcohenad@gmail.com}
\author[3]{Philip N. Klein}
\affil[3]{%
  Brown University\\
klein@brown.edu}

	\maketitle

\begin{abstract}
  The Sparsest Cut is a fundamental
  optimization problem that has been extensively studied. For planar
  inputs the problem is in $P$ and can be solved in $\tilde{O}(n^3)$ time if all vertex weights are $1$.
  Despite a significant amount of effort, the best algorithms date back to the early 90's and can only achieve $O(\log n)$-approximation in $\tilde{O}(n)$
  time or a constant factor approximation in $\tilde{O}(n^2)$ time [Rao, STOC92].
  Our main result is an $\Omega(n^{2-\eps})$ lower bound for Sparsest Cut even in planar graphs with unit vertex weights, under
  the $(min,+)$-Convolution conjecture, showing that approximations are inevitable in the near-linear time regime.
  To complement the lower bound, we provide a constant factor approximation
  in near-linear time, improving upon the 25-year old result of Rao in both time and accuracy.
  
  Our lower bound accomplishes a repeatedly raised challenge by being \emph{the first fine-grained lower bound} for a natural planar graph problem in P.
  Moreover, we prove near-quadratic lower bounds under SETH for variants of the closest pair problem in planar graphs, and use them to show that the popular Average-Linkage procedure for Hierarchical Clustering cannot be simulated in truly subquadratic time.

  At the core of our constructions is a diamond-like gadget that also settles the complexity of Diameter in \emph{distributed} planar networks. 
  We prove an $\Omega(n/\log{n})$ lower bound on the number of communication rounds required to compute the weighted diameter of a network in the CONGEST model, even when the underlying graph is planar and all nodes are $D=4$ hops away from each other. This is the first $\poly(n) + \omega(D)$ lower bound in the planar-distributed setting, and it complements the recent $poly(D, \log{n})$ upper bounds of Li and Parter [STOC 2019] for (exact) unweighted diameter and for ($1+\eps$) approximate weighted diameter.

\end{abstract}

\section{Introduction}

\paragraph{Cuts in Planar Graphs.}

The Sparsest Cut problem is among the most fundamental optimization problems.
It is NP-hard and one of the most important problems in the field of approximation algorithms; through the years, it has led to the design of new, powerful algorithmic techniques (e.g. the $O(\sqrt{\log n})$-approximation of Arora, Rao, and Vazirani \cite{arora2009expander}), and is also increasingly becoming a keystone of divide-and-conquer strategies for a variety of problems arising in graph compression \cite{dhulipala2016compressing}, clustering \cite{Das:2016,Cohen-AddadKMM19,CharikarC17}, and beyond.
The goal is to cut the graph into two roughly balanced parts while cutting as few edges as possible.

This paper studies this problem in \emph{planar graphs} where it is solvable in (weakly) polynomial time, in part because the cut can be shown to be a cycle in the dual of the graph. 
Let us mention two motivating reasons.
First, finding sparse cuts in planar graphs is of high interest
in applications such as network evaluation~\cite{DBLP:conf/wea/SchildS15}, VLSI design~\cite{bhatt1984framework,leiserson1980area},
and more~\cite{lipton1977applications}. 
For instance, sparse cuts are used to identify portions
of road networks that may suffer from congestion, or to design good VLSI layouts.
A second motivation comes from
finding \emph{optimal} separators, a ubiquitous subtask in planar graph algorithms.
The classic result of Lipton and Tarjan~\cite{lipton1979separator}
shows that any bounded-degree
planar graph has a balanced vertex separator of size $O(\sqrt{n})$, but this bound may be suboptimal in many cases (see for example~\cite{DBLP:conf/wea/SchildS15}). 
In such non-worst-case instances, algorithms for finding better separators could speed up many algorithms.

There are two common ways to define the value (or sparsity) of a cut: we divide its cost by either the weight of the smaller of the two sides, or their product.
The former definition is more standard and easier to work with in planar graphs; we will refer to it as MQC, defined below.
We will refer to the other one, that asks for a cut $S$ that minimizes $\frac{\text{cost}(S)}{ w(S) \cdot w(V-S)}$, simply as Sparsest Cut.

\begin{definition}[The Minimum Quotient Cut Problem (MQC)]
Given a graph $G=(V,E)$ with edge costs $c : E \mapsto \R^+$
and vertex weights $w : V \mapsto \R^+$, find the cut $(S,V-S), S \subseteq V$ minimizing the quotient:
$$
\text{quotient}(S) := \frac{\text{cost}(S)}{\min \{ w(S),w(V-S) \}}
$$  
where $\text{cost}(X) := \sum_{(u,v) \in E\atop u \in X, v \notin X} c(u,v)$ and $w(X)=\sum_{u \in X} w(u)$.
\end{definition}

The study of Sparsest Cut in planar graphs
dates back to Rao's first paper in the 80's \cite{Rao87}, and to subsequent works by Rao~\cite{Rao92}
and by Park and Phillips~\cite{ParkPhillips}. 
The first \emph{exact} algorithm was by Park and Phillips and had a running time of $\tilde{O}(n^2 W)$ where $W$ is the sum of the vertex weights. Here the $\tilde{O}(\cdot)$ notation hides logarithmic factors in $n$, $W$, and $C$ the sum of edge costs.
Note that in the ``unweighted'' case of unit vertex weights, $W=n$ and this upper bound is $\tilde{O}(n^3)$.
They also showed that the problem is weakly NP-Hard, and therefore cannot be solved exactly in $\tilde{O}(poly(n))$ time.
Rao's work gave a $3.5$-approximation for MQC in $\tilde{O}(n^2)$ time, and an $O(\log{n})$ approximation in $\tilde{O}(n)$ time.

Since then, there has been progress on related problems such as the Minimum Bisection problem where we want to find the cut of minimum cost that is balanced, i.e. $w(S)=w(V-S)=W/2$.
Rao~\cite{Rao87} showed that an approximation algorithm for MQC can be used to approximate Minimum Bisection in the following bi-criteria way: we return a cut with at most two-thirds of the
vertex weight on each side and with cost that is $O(\log n)$ times
the cost of the minimum bisection. 
Garg et al.~\cite{GargSV99} gave a different algorithm with similar bicriteria guarantees where the cost is only a factor of $2$ away from the optimal. This
involves iterative application of the exact algorithm of Park and Phillips.
More recently, Fox et al.~\cite{DBLP:conf/stoc/FoxKM15} gave a polynomial-time bicriteria
approximation scheme for Minimum Bisection but the algorithm runs in time
$n^{\text{poly}(1/\eps)}$.

Still, almost no progress has been made on Sparsest Cut since the early 90's.
In the most interesting regime of near-linear running times, Rao's $O(\log{n})$-approximation is the best known, and there is no exact algorithm running in time $\tilde{o}(n^2W)$.
The gaps are large, but the most pressing question is:

\begin{oq}
Can Sparsest Cut in planar graphs be solved \emph{exactly} in near-linear time?
\end{oq}

Given that the upper bound is longstanding, it is natural to try to use the recent tools of fine-grained complexity in order to resolve this question negatively. 
Can we show that a linear time algorithm would refute SETH or one of the other popular conjectures? 
This is challenging because this field has not been successful in proving \emph{any conditional lower bound for a planar graph problem in P}, not to mention a natural and important problem like Sparsest Cut.
Nonetheless, our main result is a quadratic conditional lower bound even for the unit-vertex-weight version of Sparsest Cut where the upper bound is cubic, and it also applies for MQC and Minimum Bisection. 
The lower bound is based on the hypothesis that the basic $(min,+)$-Convolution problem requires quadratic time.
This hardness assumption was recently highlighted by Cygan et al. \cite{DBLP:conf/icalp/CyganMWW17} after being used in other papers \cite{backurs2017better,laber2014lower,DBLP:conf/icalp/KunnemannPS17,bateni2018fast}.
It is particularly appealing because it implies both the $3$-SUM and the All-Pairs Shortest Paths conjectures, and therefore also all the dozens of lower bounds that follow from them (see \cite{DBLP:conf/icalp/CyganMWW17,virginia_ICM}).

\begin{theorem}
\label{thm:sparsestcut}
If for some $\eps>0$, the Sparsest Cut, the Minimum Quotient Cut, or the Minimum Bisection problems can be solved in $O(n^{2-\eps})$ time in planar graphs of treewidth $3$ on $n$ vertices with unit vertex-weights and total edge cost $C= n^{O(1)}$, then the $(min,+)$-Convolution problem can be solved in $O(n^{2-\eps})$ time.
\end{theorem}

After settling the high-order question it is easier to direct our energies into decreasing the gaps. 
A natural next question is whether there could be a cubic lower bound, which would completely settle the exact case.
We show that this is not the case; a natural use of $O(\sqrt{n})$-size separators in the right way inside the Park and Phillips algorithm reduces the running time to $n^{2.5}$.
Figuring out the exact exponent remains an important open question. It seems that new algorithmic techniques will be required to bring the upper bound down to $O(nW)$, yet we do not know of hard instances that seem to require super-quadratic time.

\begin{theorem}
\label{thm:sparsestcut:exact}
The Sparsest Cut and Minimum Quotient Cut problems in planar graphs on $n$ vertices with total vertex-weight $W$ and total edge costs $C$ can be solved in $O(n^{3/2}W \log(CW))$ time.
\end{theorem}

Since near-linear time algorithms are the most desirable, perhaps the next most pressing question is whether the $O(\log{n})$ approximation of Rao is the best possible:

\begin{oq} 
Is there an $O(1)$-approximation algorithm for MQC in near-linear time?
\end{oq}

We give such an algorithm.  It combines several techniques with new ideas; the main advantage comes from finding and utilizing a node that is guaranteed to be close to the optimal cycle rather than on it. 

\begin{theorem}
  \label{thm:sparsestcut:approx}
The Minimum Quotient Cut problem in planar graphs on $n$ vertices with
total vertex-weight $W$ and total edge cost $C$ can be approximated to
within an $O(1)$ factor in time $n \log^{O(1)}(CWn)$.
\end{theorem}

\paragraph{New Hardness Results in Planar Graphs.}
Theorem~\ref{thm:sparsestcut} finally resolves a repeatedly raised challenge in fine-grained complexity: \emph{Are there natural planar graph problems in P for which we can prove a conditional $\omega(n)$ lower bound?}
The list of problems with such lower bounds under SETH or other conjectures is long, exhibiting problems on graphs \cite{WW18}, strings \cite{AWW14,BI15}, geometric data \cite{GO12,Bringmann14,BDT16}, trees \cite{ABHWZ16,BGMW18}, dynamic graphs \cite{AV14,HKNS15}, compressed strings \cite{ABBK17}, and more\footnote{For a more extensive list see the survey in \cite{virginia_ICM}.}.
Perhaps the most related  results are the lower bounds for problems on \emph{dynamic} planar graphs \cite{AD16} but those techniques do not seem to carry over to the more restricted setting of (static) graphs.
Indeed, the above question has been raised repeatedly, even after \cite{AD16}, including in the best paper talk of Cabello at SODA 2017 \cite{Cabello17}.
The search for an answer to this question has been remarkably fruitful from the viewpoint of upper bounds; Cabello's breakthrough (a subquadratic time algorithm for computing the diameter of a planar graph) came following attempts at proving a quadratic lower bound (such a lower bound holds in sparse but non-planar graphs \cite{RV13}), and the techniques introduced in his work (mainly Abstract Voronoi Diagrams) have led to major breakthroughs in planar graph algorithms (see \cite{stoc19_planar,CADW17,GKMS18,GMWW18}).

Strong lower bounds were found for some restricted graph classes such as graphs with logarithmic treewidth \cite{AWW16} (e.g. a quadratic lower bound for Diameter); but these are incomparable with planar graphs.
For some problems such as subgraph isomorphism there are lower bounds even for trees \cite{ABHWZ16},  a restricted kind of planar graphs; however these problems are not in P when the graphs are planar but not trees.
Many hardness results are known for geometric problems on points in the plane (e.g. \cite{GO12,BH01,BDT16}); while related in flavor, the techniques are specific to the euclidean nature of the data and it is not clear how to extract lower bounds for natural graph problems out of these results.

The main challenge, of course, is in designing \emph{planar} gadgets and constructions that are capable of encoding fine-grained reductions. While this has already been accomplished in other contexts such as NP-hardness proofs or in parameterized complexity, those techniques do not work under the more strict efficiency requirements that are needed for fine-grained reductions.
From the perspective of lower bounds, the main contribution of this paper is in coming up with a planar construction that exhibits the super-linear complexity of basic problems like Sparsest Cut. 
By extracting the core gadget from this construction and building up on it we are able to prove lower bounds for other, seemingly unrelated problems on planar graphs.

Notably, our constructions are not only planar but also have very small treewidth of two or three, but crucially not one since our problems become easy on trees. This might be of independent interest. 

\paragraph{Closest Pair of Sets and Hierarchical Clustering.} 
Hierarchical Clustering (HC) is a ubiquitous task in data science and machine learning. Given a data set of $n$ points with some similarity or distance function over them (e.g. points in Euclidean space, or the nodes of a planar graph with the shortest path metric), the goal is to group similar points together into clusters, and then recursively group similar clusters into larger clusters. 
Perhaps the two most popular procedures for HC are Single-Linkage and Average-Linkage.
Both are so-called \emph{agglomerative} HC algorithms (as opposed to \emph{divisive}) since they proceed in a bottom-up fashion:
In the beginning, each data point is in its own cluster, and then the most similar clusters are iteratively merged - creating a larger cluster that contains the union of the points from the two smaller clusters - until all points are in the same, final cluster.

The difference between the different procedures is in their notion of similarity between clusters, which determines the choice of clusters to be merged.
In Single-Linkage the distance (or \emph{dissimilarity}) is defined as the minimum distance between any two points, one from each cluster.
While in Average-Linkage we take the average instead of the minimum.
It is widely accepted that Single-Linkage is sometimes simpler and faster, but the results of Average-Linkage are often more meaningful. 
Extensive discussions of these two procedures (and a few others, such as Complete-Linkage where we take the max, rather than min or average) can be found in many books (e.g. \cite{friedman2001elements,leskovec2014mining,IRbook,handbook}), surveys (e.g. \cite{murtagh83,Murtagh92comments,Carlsson10}), and experimental studies (e.g. \cite{scikit-learn}). 

Both of these procedures can be performed in nearly quadratic time and a faster, subquadratic implementation is highly desirable. Some subquadratic algorithms that try to approximate the performance of these procedures have been proposed, e.g. \cite{CM15,our_neurips}.
However, it is often observed that an exact implementation is at least as hard as finding the closest pair (of data points), since they are the first pair to be merged.
Indeed, if the points are in Euclidean space with $\omega(\log{n})$ dimensions, the Closest Pair problem requires quadratic time under SETH \cite{AW15,SM19}, and therefore these procedures cannot be sped up without a loss.

But what if we are in the planar graph metric? 
This argument breaks down because the Closest Pair problem is trivial in planar graphs (the minimum weight edge is the answer).
Moreover, the Single-Linkage procedure can be implemented to run in near-linear time in this setting, since it reduces to the computation of a minimum spanning tree \cite{MST_SL69}.
In fact, subquadratic algorithms are known for many other metrics that have subquadratic closest pair algorithms such as spaces with bounded doubling dimension \cite{march2010fast}, and efficient approximations are known when the closest pair can be approximated efficiently \cite{our_neurips}.
This naturally leads to the question:

\begin{oq}
Can Average-Linkage be computed in subquadratic time in any metric where the closest pair can be computed in subquadratic time? 
\end{oq}

Surprisingly to us, it turns out that the answer is no. 
In this paper we prove a near-quadratic lower bound under SETH for simulating the Average-Linkage and Complete-Linkage procedures in planar graphs, by proving a lower bound for variants of the closest \emph{pair of sets} problem which are natural problems of independent interest: 
We are given a planar graph on $n$ nodes that are partitioned into $O(n)$ sets and the goal is to find the pair of sets that minimizes the sum (or max) of pairwise distances.
An $O(n^2)$ upper bound is easy to obtain from an all-pairs shortest paths computation.

\begin{theorem}
\label{thm:HC}
If for some $\eps>0$, the Closest Pair of Sets problem, with sum-distance or max-distance, in unweighted planar graphs on $n$ nodes can be solved in $O(n^{2-\eps})$ time, then SETH is false.
Moreover, if for some $\eps>0$ the Average-Linkage or Complete-Linkage algorithms on $n$ node planar graphs with edge weights in $[O(\log{n})]$ can be simulated in $O(n^{2-\eps})$ time, then SETH is false.
\end{theorem}

\paragraph{Diameter in Distributed Graphs.}
Our final result is on the complexity of diameter in planar graphs in the CONGEST model.
This is the central theoretical model for \emph{distributed computation}, where the input graph defines the communication topology: in each round, each of the $n$ nodes can send an $(\log{n})$-bit message to each one of their neighbors. The complexity of a problem is the worst case number of rounds until all nodes know the answer.

In the CONGEST, a problem is considered tractable if it can be solved in $poly(D,\log{n})$ time, where $D$ is the diameter of the underlying unweighted network\footnote{Note that $\Omega(D)$ rounds are usually required; some nodes cannot exchange any information otherwise.} (i.e. the hop-diameter).
Many basic problems such as finding a Minimum Spanning Tree (MST) and distance computations have been shown to be intractable \cite{PelegR99,Elkin06,NSP11,SarmaHKKNPPW12,FrischknechtHW12,ACK16,CKP17,BK18}.
For example, no algorithm can decide whether the diameter of the network is $3$ or $4$ in $n^{o(1)}$ rounds \cite{FrischknechtHW12}. That is, the Diameter problem itself cannot be solved in $poly( \text{diameter} )$ time.

While (sequential\footnote{This seems to be the standard term for \emph{not distributed} algorithms.}) algorithms for planar graphs have been an extensively studied subject for the past three decades, only recently have they been considered in the distributed setting \cite{Li18,HIZ16a,HIZ16b,HHW18,HLZ18,GP17}.
This study was initiated by Ghaffari and Hauepler \cite{GhaffariH16a,GhaffariH16b} who also demonstrated its potential: While MST has an $\Omega(\sqrt{n})$ lower bound in general graphs~\cite{SarmaHKKNPPW12}, the problem is tractable on planar graphs.
All previous lower bound constructions are far from being planar, and it is natural to wonder: Do \emph{all} problems\footnote{Here, we mean \emph{decision} problems. It is easy to show that problems with a large output such as All-Pairs-Shortest-Paths are not tractable even in trivial networks.} become tractable in the CONGEST when the network is planar?\footnote{Note that even NP-Hard problems might become tractable in this model, since the only measure is the number of rounds, not the computation time at the nodes. For example, in the LOCAL model where we do not restrict the messages to be short, all problems can be solved in $O(D)$ rounds.}

In this paper, we provide a negative answer with a simple argument ruling out any $f(D) \cdot n^{o(1)}$ distributed algorithms even in planar graphs.
A very recent breakthrough of Li and Pater showed that the diameter problem in \emph{unweighted planar graphs} is tractable in the CONGEST \cite{stoc19_parter}.
We show that the \emph{weighted} case is intractable.
Our lower bound is only against exact algorithms which is best-possible since Li and Parter achieve a $(1+\eps)$-approximation in the weighted case with $\tilde{O}(D^6)$ rounds.

\begin{theorem}
\label{thm:diam}
	The number of rounds needed for any protocol to compute the diameter of a weighted planar network of constant hop-diameter $D=O(1)$ on $n$ nodes  in the $CONGEST$ model is $\Omega(\frac{n}{\log{n}})$.

\end{theorem}

Our technique for showing lower bounds in the CONGEST model is by reduction from two-party communication complexity and is similar to the one in previous works.
For general graphs, there are strong $\Omega(n/\log^{2}n)$ lower bounds for computing the diameter even in unweighted, sparse graphs of constant diameter \cite{FrischknechtHW12,ACK16,BK18}.
Our high-level approach is similar, but a substantially different implementation is needed in order to keep the graph planar.
In fact, we design a simple but subtle, diamond-like gadget for this purpose (see Section~\ref{sec:diam}).
The other lower bounds in the paper were obtained by building on top of this simple construction and they show that this gadget may really be capturing the difficulty in many planar graph problems.
In particular, the lower bounds for closest pair of sets, which are the most complicated in this paper, are achieved by combining $O(n)$ copies of this gadget together in an ``outer construction'' that also has the same structure of this gadget.

\section{Sparsest Cut, Minimum Quotient Cut and Minimum Bisection}
\label{sec:def:sc}
We now provide formal definitions of the Sparsest Cut, Minimum Quotient
Cut and Minimum Bisection problems.
Consider a planar graph $G=(V,E)$ with edge costs $c : E \mapsto \R^+$
and vertex weights $w : V \mapsto \R^+$. Given a subset of vertices $S$,
we define the \emph{cut induced by $S$} as the set of edges with one extremity
in $S$ and the other in $V-S$. We will slightly abuse notation by
referring to the cut induced by $S$ as the cut of $S$.
We let $$\text{cost}(S) = \sum_{(u,v) \in E\atop u \in S, v \notin S} c(u,v)$$
and, with a slight abuse of notation, $w(S) = \sum_{v \in S} w(v)$.
Given a subset of vertices $S$, we define
the \emph{sparsity of the cut induced by $S$} as the ratio
$\cut(S)/(w(S) \cdot w(V-S))$. 
The \emph{Sparsest Cut problem}
asks for a subset $S$ of $V$ that has minimum sparsity over all cuts induced
by a subset $S \subset V$.
This is not to be confused with the General
Sparsest Cut  which is APX-Hard in planar graphs\footnote{There, there is a weighted demand between
pair of vertices and the goal is to find a subset $S$ such that the cut induced
by $S$ minimizes the ratio of $\text{cut}(S)$ to the amount of demand between
pairs of vertices in $S$ and $V-S$.}.

The \emph{quotient} of a cut $S$ is defined to be $\cut(S)/(\min\{w(S),w(V-S)\})$.
The \emph{Minimum Quotient Cut}
problem asks for a cut with minimum quotient.  The \emph{Minimum Bisection problem}
asks for a subset $S$ such that $w(S) = w(V)/2$ and that minimizes
$\cut(S)$. 


\subsection{Proof of Theorem~\ref{thm:sparsestcut}: Lower bounds}


In this section, we aim prove a conditional lower bound of $\Omega(n^{2-\eps})$ for the
unit vertex-weight case of all three problems: the Sparsest Cut, the Minimum Quotient Cut,
and the Minimum Bisection problems.
We will first provide a reduction for the case of non-unit
vertex-weight and then show how to
adapt it to the unit vertex-weight case.

Our lower bounds are based on the hardness of the $(min,+)$-Convolution Problem, defined as follows, which is conjectured to require $\Omega(n^{2-\eps}$ time, for all $\eps>0$.

\begin{definition}[The $(min,+)$-Convolution Problem]
Given two sequences of integers $A = \langle a_1,\ldots,a_n\rangle$,
$B = \langle b_1,\ldots,b_n \rangle$, the output is a sequence $C = \langle c_1,c_2,\ldots,c_n \rangle$
such that $c_k = \min_{0 \le i \le k}  a_{k-i+1} + b_i$.
\end{definition}

To prove our conditional lower bounds we will show reductions from the following variant called $(min,+)$-Convolution Upper Bound, which was shown to be \emph{subquadratic-equivalent} to $(min,+)$-Convolution by Cygan et al.~\cite{DBLP:conf/icalp/CyganMWW17}.
Namely, there is an
$O(n^{2-\eps})$ algorithm for some constant $\eps>0$
for the $(min,+)$-Convolution Upper Bound problem
if and only if there is an $O(n^{2-\eps'})$ for the $(min,+)$-Convolution problem, for some $\eps'>0$.

\begin{definition}[The $(min,+)$-Convolution Upper Bound Problem]
Given three sequences of positive integers $A = \langle a_1,a_2,\ldots,a_n\rangle$,
$B = \langle b_1,b_2,\ldots,b_n \rangle$, and
$C = \langle c_1,c_2,\ldots,c_n \rangle$, verify that for all $k \in [n]$, there is no pair $i,j$
such that $i+j = k$ and
$a_i + b_j < c_k$.
\end{definition}

\paragraph{The Reduction.}
The construction in each of our three reductions is the same, and the analysis is a little different in each case. Therefore, we will present all three in parallel. 
Given an instance $A,B,C$ of the $(min,+)$-Convolution Upper Bound problem, we build
an instance $G$ for the Sparsest Cut, Minimum Quotient Cut or
Minimum Bisection problems as follows.

Let $T =  \sum_{i=1}^n a_i + b_i +c_i$ and $\beta = 4Tn^2$.
The graph $G$ will have two special vertices $u$ and $v$ of weights $10n$ and $11n$ respectively.
It will also have three paths $P_A,P_B,P_C$ that connect $u$ and $v$ and will encode the three sequences as follows.

\begin{itemize}
\item The path $P_A$ has a vertex $v_{a_i}$ for
each $a_i$ in $A$, of weight $1$.  We connect $v_{a_i}$ to $v_{a_{i+1}}$
with an edge
of cost $\beta + a_i$ for each $1 \le i < n$. Moreover, we connect $v_{a_n}$ to $u$ with an edge
of cost $\beta + a_n$ and $v$ to $v_{a_1}$ with an edge $e_A$ of cost $1210 n^2 (2\beta+ T)$.

\item The path $P_B$ is defined in an analogous way.
We create a vertex $v_{b_i}$ of weight
$1$ for each $b_i$ in $B$ and connect it with an edge of cost $\beta+b_i$ to $v_{b_{i+1}}$
for each $1 \le i < n$, and we connect $v_{b_n}$ to $u$ with an edge of cost
$\beta + b_n$ and $v$ to $v_{b_1}$ with an edge $e_B$ of cost $1210 n^2 (2\beta+ T)$.

\item The path $P_C$ is defined differently: the indices are ordered in the opposite direction and the numbers are flipped. We create
a vertex $v_{c_i}$ of weight $1$ for each $c_i$ in $C$ but connect
each $v_{c_i}$ to $v_{c_{i+1}}$ with an edge of cost $\beta+T- c_i $ for each
$1 \le i < n$. And this time we connect $v_{c_1}$ to $u$ with an edge $e_C$ of cost $1210 n^2 (2\beta+ T)$ and $v_{c_n}$ to $v$ with an edge of cost $\beta+c_n$.

\end{itemize}

It is easy to see that the resulting graph is planar and has
treewidth at most 3. See also Figure~\ref{fig:LB}.
The total weight in our construction is $W=24n$ because there are $3n$ vertices of weight $1$ and the two special vertices $u,v$ have weight $21n$.

\begin{figure}
  \centering

  \def\svgwidth{7cm}
  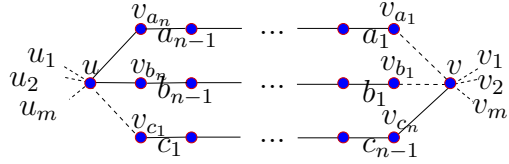
  \caption{The graph generated in our reduction. Dashed edges
  are edges of weight $1210 n^2 (2\beta+ T)$.}
  \label{fig:LB}
\end{figure}

\paragraph{Correctness of the Reductions.}
To analyze the reduction, we start by proving two lemmas about the structure of the optimal solution in each of the three problems in the instances we generate.
To build intuition, observe that in our construction any cut that does not separate $u$ and $v$ is far from being balanced and therefore will not be an optimal solution.
Another observation is that the weights of the edges $\{e_A,e_B,e_C\}$ is practically infinite and therefore they will not be cut by an optimal cut.

\begin{lemma}
  \label{lem:structLB:sparsest}
  The Minimum Quotient Cut, the Sparsest Cut, and the Minimum Bisection
  cut intersect each of $P_A,P_B,P_C$ exactly once and do not intersect
  any edge of $\{e_A,e_B,e_C\}$.
\end{lemma}
\begin{proof}
  We start with the Minimum Bisection, which is the simplest case since the cut is forced to have exactly $W/2$ weight on each side. By picking
  edges $(v_{a_{1}},v_{a_{2}})$, $(v_{b_{1}},v_{b_{2}})$,
  and $(v_{c_{2}},v_{c_{3}})$,
  we indeed obtain a cut that breaks the graph into two connected components
  of the same weight. The value of the cut is then at most
  $\beta+a_1 + \beta+b_1 + \beta + T - c_2 \le 3\beta+2T$.
  However, any cut intersecting $\{e_A,e_B,e_C\}$ has cost at least
  $1210n^2 (2\beta + T)$ and so 
  the (optimal) Minimum Bisection does not intersect $\{e_A,e_B,e_C\}$.
  Moreover, it is easy to see that the Minimum Bisection Cut must
  separate $u$ from $v$ as otherwise, the cut is not balanced.
  Thus the Minimum Bisection intersects
  each of $P_A,P_B,P_C$ at least once. Finally, suppose it intersects
  them more than once. The cost is thus at least
  $4 \beta$, while by picking
  edges $(v_{a_{1}},v_{a_{2}})$, $(v_{b_{1}},v_{b_{2}})$,
  and $(v_{c_{2}},v_{c_{3}})$, the cost achieved is at most $2T+3\beta$.
  By the choice of $\beta$, we have $2T+3\beta < 4\beta$ and so the
  Minimum Bisection intersects each of $P_A,P_B,P_C$ exactly once.
  
  We then argue that the Minimum Quotient cut $Q$ and the
  Sparsest Cut $S$ do not intersect any
  edge of $\{e_A,e_B,e_C\}$. Indeed, any cut $U$ that intersect
  an edge $\{e_A,e_B,e_C\}$
  has cost at least $1210 n^2 (2\beta+ T)$ and so induces
  a Quotient Cut of value at least
  $1210 n^2 (2\beta+ T)/(12n)   = 110 n (2\beta+T)$ and a cut of
  Sparsity at least
  $1210 n^2 (2\beta+ T)/(12n)^2 = 10 (2\beta+T)$.
  Now, consider the cut separating $a_1$ from the rest of the graph.
  This cut has cost at most $2\beta + T$.
  Thus, it forms a quotient cut of value at most $2\beta+ T$ and a
  cut of sparsity at most
  $(2\beta+T)/12n$. This induces a cut that is both of
  smaller sparsity and of smaller
  quotient value than any cut
  involving any of $\{e_A,e_B,e_C\}$.
  It follows that $Q$ and $S$ do not intersect $\{e_A,e_B,e_C\}$.

  We now show that both $Q$ and $S$ separate $u$ from $v$.
  Consider a cut $U$ that has both $u$ and $v$ on one side. This cut needs to contain
  at least two edges and so has cost at least $2(\beta + 1)$.
  It thus induces a quotient cut of value at least $2(\beta+1)/3n$ and a cut of
  sparsity at least $2(\beta+1)/(22n \cdot 3n)$.
  On the other hand, consider a cut $Y$ obtained by picking an edge from each of $P_A, P_B, P_C$.
  The cost of this cut is at most $3\beta+2T$, which induces a quotient cut
  of value at most $(3\beta+2T)/(10n)$ and of sparsity $(3\beta+2T)/(10n)^2$.
  Since $\beta = 4Tn^2$, we have that
  $(3\beta+2T)/(10n) < (2\beta+1)/(3n)$ and
  $(3\beta+2T)/(10n)^2 < 2(\beta+1)/(20n \cdot 3n)$, as long as $n > 1$. Therefore, 
  $Q$ and $S$ separate $u$ from $v$ and so intersect
  at least one edge from each of $P_A,P_B,P_C$.

  Finally, by Theorem 2.2 in~\cite{ParkPhillips} and Proposition 2.3 in~\cite{patel13},
  we have that the minimum quotient cut and the sparsest cut
  are simple cycles in the dual of the graph. Picking two edges of $P_A$ (or of $P_B$, or $P_C$)
  together with at least one edge of $P_B$ and of $P_C$ would induce a non-simple cycle in the dual
  of the graph and so a non-optimal cut. Therefore, we conclude that 
  the minimum quotient cut and sparsest cut uses exactly one edge of $P_A$, one edge of $P_B$,
  and one edge of $P_C$.
\end{proof}

\begin{lemma}
  \label{lem:LB:cut}
  If the Minimum Quotient Cut, the Sparsest Cut, or the
  Minimum Bisection intersects edges
  $(v_{a_i},v_{a_{i+1}})$,
  $(v_{b_j},v_{b_{j+1}})$ and $(v_{c_k},v_{c_{k+1}})$, then $v$ and the vertices in
  $\{v_{a_1},\ldots, v_{a_i}\}$, $\{v_{b_1},\ldots,v_{b_j}\}$, and
  $\{v_{c_{k+1}},\ldots,v_{c_n}\}$, are on one side of the cut while
  the remaining vertices are on the other side.
\end{lemma}
\begin{proof}
  By Lemma~\ref{lem:structLB:sparsest},
  the Minimum Quotient Cut, the Sparsest Cut and the Minimum Bisection
  intersect each of $P_A,P_B,P_C$ exactly once. Thus, if
  one of them intersect edges
  $(v_{a_i},v_{a_{i+1}})$, $(v_{b_j},v_{b_{j+1}})$ and $(v_{c_k},v_{c_{k+1}})$, then
  $v_{a_i}$ remains connected to $v$ through the path
  $\{v_{a_1},\ldots, v_{a_i}\}$ and
  so all the vertices in $\{v, v_{a_1},\ldots, v_{a_i}\}$ are in the
  same connected
  component. The remaining vertices of $P_A$ remains connected to $u$.
  A similar reasoning applies to $v_b$ and $v_c$ and yields the lemma.
\end{proof}

From these two lemmas it follows that the only way that an optimal cut can be completely balanced (i.e. has weight $W/2=12n$ on each side) is by cutting three edges $(v_{a_i},v_{a_{i+1}})$,
  $(v_{b_j},v_{b_{j+1}})$ and $(v_{c_k},v_{c_{k+1}})$, where $i+j=k$. 
  This is the crucial property of our construction. To see why it is true, note that $i+j+(n-k)$ vertices go to the side of $v$ while $(n-i)+(n-j)+k$ vertices go to the side of $u$, and so to achieve balance it must be that:
  $$
  i+j-k + n + w(v) =  k - i - j + 2n + w(u) 
  $$
  which simplifies to $ i+j =k  $ because of our choice of $w(u)=10n$ and $w(v)=11n$.
  Moreover, the cost of this cut is exactly $(3\beta + T) + (a_i + b_j - c_k)$ which is less than $(3\beta+T)$ if and only if $a_i+b_j < c_k$.
The correctness of the reductions follows from the following claim.

\begin{claim}
There is no $k \in [n]$ and a pair $i,j$ such that $i+j=k$ and $a_i+b_j < c_k$, if and only if either of the following statements is true:
\begin{itemize}
\item the Minimum Quotient Cut has value at least $(3\beta+T)/12n$,
\item the Sparsest Cut has value at least $(3\beta+T)/(12n^2)$, or
\item the Minimum Bisection has value at least $(3\beta+T)$.
\end{itemize}
\end{claim}
%
\begin{proof}
Consider first the Minimum Bisection. By
Lemma~\ref{lem:structLB:sparsest}, the Minimum Bisection intersects
each of $P_A,P_B,P_C$ exactly once. Thus, combined with
Lemma~\ref{lem:LB:cut}, we have that the if the Minimum Bisection
intersects an edge $(v_{c_k},v_{c_{k+1}})$ for some $k$, then it must 
intersect $(v_{a_i},v_{a_{i+1}})$, and $(v_{b_j},v_{b_{j+1}})$ such
that $j+i = k$ to achieve balance. Therefore, the cut has value
$3\beta + a_i + b_{k-i} + T-c_k$ which is at least $3\beta+T$
if and only if there is no
$i,j$ such that $i+j = k$ and $a_i+b_j < c_k$.

We now turn to the cases of Minimum Quotient Cut and Sparsest Cut.
For the first direction, assume that there is a triple $i,j,k$ where $k=i+j$ such that $a_i+b_j < c_k$.
In this case, we have
a cut of quotient value less than $(3\beta+T)/(12n)$ and a cut of sparsity less than
$(3\beta+T)/(12n)^2$ obtained by taking edges $(v_{a_i},v_{a_{i+1}})$,
$(v_{b_j},v_{b_{j+1}})$ and $(v_{c_k},v_{c_{k+1}})$.

For the other direction, let us first focus on the Minimum Quotient
Cut $Q$.
By Lemma~\ref{lem:structLB:sparsest},
$Q$ contains one edge from each of $P_A,P_B,P_C$ say $(v_{a_i},v_{a_{i+1}})$,
$(v_{b_j},v_{b_{j+1}})$ and $(v_{c_k},v_{c_{k+1}})$.
First, if $i+j \neq k$, by Lemma~\ref{lem:LB:cut}, we have that
the cut has quotient value at least $(3\beta + a_i + b_j + T-c_k)/(12n -1)$
which is at least $(3\beta + 1)/(12n -1)$.
By the choice of $\beta$, we have
that $3\beta/12n > 10T$ and so, $(3\beta + 1)/(12n -1) \ge
(3\beta + T)/12n$.

Thus, we may assume that $i+j = k$.
By Lemma~\ref{lem:LB:cut}, we hence have that the quotient value of the cut
is less than $(3\beta+T)/(12n)$ if and only if $a_i + b_j < c_k$.
This follows from the fact that 
the quotient value of the cut is $(3\beta+a_i + b_i + T-c_k)/(12n)$
which is less than $ (3\beta + T)/(12n)$
if and only if $a_i + b_j < c_k$.



The argument for the Sparsest Cut is similar. Again, by
Lemma~\ref{lem:LB:cut},
the sparsest cut contains one edge from each of $P_A,P_B,P_C$,
say $(v_{a_i},v_{a_{i+1}})$, $(v_{b_j},v_{b_{j+1}})$ and $(v_{c_k},v_{c_{k+1}})$.
Similarly, if $i+j = k$, we have that the sparsity of the cut
is less than $(3\beta+T)/(12n)^2$ if and only if $a_i + b_j < c_k$,
since the sparsity of the cut is $(3\beta+a_i + b_i + T-c_k)/(12n)^2$.

Finally, if $i+j \neq k$ then the sparsity of the cut is at least
$(3\beta + a_i + b_j + T-c_k)/((12n -1)(12n+1))$ which
is at least $(3\beta + 1)/((12n -1)(12n+1))$.
By the choice of $\beta$, we have that
$3\beta/((12n)^2-1) > 10T$ and so, $(3\beta + 1)/(12n -1) \ge
(3\beta + T)/12n$.
\end{proof}

\paragraph{A Unit-Vertex-Weight Reduction}
Intuitively, we are able to remove the weights because the total weight $W$ is $O(n)$.
To show this more precisely, we note that the above reduction makes use of vertices of weight $1$, except for $u$ and $v$ which
are of weight $10n$ and weight $11n$ respectively. Now, place a weight of $1$ on $u$ and $v$ and
add vertices $u^1,\ldots,u^{10n-1}$ and connect them with edges of length $1210 n^2 (2\beta+T)$
to $u$ and add vertices $v^1,\ldots,v^{11n-1}$ and connect them with edges of length
$1210 n^2 (2\beta+T)$ to $v$. For the same argument used in Lemma~\ref{lem:structLB:sparsest},
the Minimum Quotient cut, the Sparsest Cut, and the Minimum Bisection do not intersect any of these edges and so
the above proof can be applied unchanged.

\section{Lower Bound for Diameter in CONGEST}
\label{sec:diam}

In this section we prove Theorem~\ref{thm:diam} and present the simple gadget that is at the core of our lower bounds.

\begin{proof}[Proof of Theorem~\ref{thm:diam}]

The proof is by reduction from the two-party communication complexity of Disjointness: There are two players, Alice and Bob, each has a private string of $n$ bits, $A,B \in \{0,1\}^n$ and their goal is to determine whether the strings are disjoint, i.e. for all $i \in [n]$ either $A[i]=0$ or $B[i]=0$ (or both).
It is known that the two players must exchange $\Omega(n)$ bits of communication in order to solve this problem \cite{Razborov92}, even with randomness, and we will use this lower bound to derive our lower bound for distributed diameter.

Let $A,B\in \{0,1\}^n$ be the two private input strings in an instance $(A,B)$ of Disjointness. We will construct a planar graph $G$ on $O(n)$ nodes based on these strings and show that a CONGEST algorithm that can compute the diameter of $G$ in $T(n)$ rounds implies a communication protocol  solving the instance $(A,B)$ in $O(T(n) \log{n})$ rounds. This is enough to deduce our theorem.

The nodes $V$ of $G$ are partitioned into two types: nodes $V_A$ that ``belong to Alice'' and nodes $V_B$ that ``belong to Bob''.
For each coordinate $i \in [n]$ we have two nodes $a_i \in V_A$ and $b_i \in V_B$.
In addition, there are four special nodes: $\ell,r \in V_A$ and $\ell',r' \in V_B$.
In total, there are $|V_A|+|V_B|=2n+4$ nodes in $G$.

Let us first describe the edges $E$ of $G$ before defining their weights $w:E\to G$.
The edges are independent of the instance $(A,B)$ but their weights will be defined based on the strings.
Every coordinate node $a_i$, for all $i \in [n]$, has two edges: one left-edge (which will be drawn to the left of $a_i$ in a planar embedding) connecting it to $\ell$, and one right-edge connecting it to $r$.
Similarly for Bob's part of the graph, every coordinate $b_i$ has a left-edge to $\ell'$ and a right-edge to $r'$.
Finally, there is an edge connecting $\ell$ with $\ell'$ and an edge connecting $r$ with $r'$.

One way to embed $G$ in the plane is as follows:
The nodes $a_1,\ldots,a_n,b_1,\ldots,b_n$ are ordered in a vertical line with $a_1$ at the top. 
In between $a_n$ and $b_1$ we add some empty space in which we place the other four nodes in $G$ such that $\ell,\ell'$ are to the left of the vertical line and $r,r'$ are to the right, and
 the four nodes are placed in a rectangle-like shape with $\ell,r$ on top and $\ell',r'$ on the bottom.

The final shape (see Figure\ref{fig:diam}) looks like a diamond (especially if we rotate it by 90 degrees) with $\ell,\ell'$ on top and $r,r'$ on the bottom.
It is important to observe that the hop-diameter $D$ of this graph is a small constant, $D=3$.
A crucial property of $G$ for the purposes of reductions from two-party communication problems is that there is a very small cut between Alice's and Bob's parts of the graph: there are only two edges that go from one part to the other ($(\ell,\ell')$ and $(r,r')$).

The main power of this gadget comes from the weights, defined next.
Set $M=4$ (but it will be useful to think of $M$ as a large weight).
\begin{align*}
w(\ell, \ell') &= M \\
w(r, r') & = M \\
w(a_i, \ell) & =     \begin{cases}
      i \cdot M, & \text{if}\ A[i]=0 \\
      i \cdot M + 1, & \text{if}\ A[i]=1 
    \end{cases}
    \\
w(a_i, r) & =     \begin{cases}
      (n+1 - i) \cdot M, & \text{if}\ A[i]=0 \\
      (n+1-i) \cdot M + 1, & \text{if}\ A[i]=1 
    \end{cases}
    \\
w(b_j, \ell') & =     \begin{cases}
      (n+1 - j) \cdot M, & \text{if}\ B[j]=0 \\
      (n+1-j) \cdot M + 1, & \text{if}\ B[j]=1 
    \end{cases}
	\\
w(b_j, r') & =     \begin{cases}
      j \cdot M, & \text{if}\ B[j]=0 \\
      j \cdot M + 1, & \text{if}\ B[j]=1 
    \end{cases}
\end{align*}

\begin{figure}
	\centering
	\scalebox{0.5}{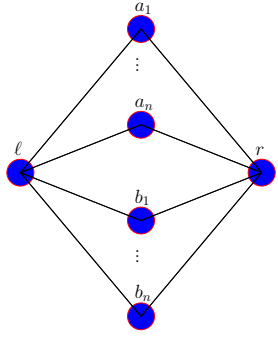}
	\caption{Our basic construction. For the Diameter CONGEST lower bound, the nodes $\ell$ and $r$ are each split into an edge. The complexity of handling this gadget comes from a careful choice of the weights that makes $a_i$ and $b_i$ ``interact'' for all $i \in [n]$, while the other pairs are not effective.}
	\label{fig:diam}
\end{figure}

The key property of this construction is that every pair of nodes in $G$ will have distance less than $ (n+2) \cdot M$ \emph{except for pairs $a_i,b_j$ with $i=j$.} 
And for these special pairs $a_i,b_i$ the distance will be exactly $ (n+2) \cdot M$ plus $0$,$1$, or $2$, depending on $A[i],B[i]$; thus the diameter of $G$ will be affected by whether $A,B$ are disjoint.
Achieving this kind of property is the crux of most reductions from Disjointness to graph problems.
Next we formally show such bounds on the distances in $G$.

\begin{claim}
\label{cl:gadget_main}
The weighted diameter of $G$ is $ (n+2) \cdot M+2$ if there exists an $i\in [n]$ such that $A[i]=B[i]=1$ and it is at most $ (n+2) \cdot M + 1$ otherwise.
\end{claim}
\begin{proof}
The proof is by a case analysis on all pairs of nodes $x,y$ in $G$.
We start with the less interesting cases, and the final case is the interesting one (which will depend on $A,B$).
\begin{itemize}
\item If $x=a_i$ and $y \in \{\ell,\ell',r,r'\}$ then the path of length one or two from $x$ to $y$ has weight $d(x,y) \leq n \cdot M +1 + M = (n+1) \cdot M+1$.
\item Similarly for Bob's side, if $x=b_j$ and $y \in \{\ell,\ell',r,r'\}$ then $d(x,y) \leq n \cdot M +1 + M = (n+1) \cdot M+1$.

\item If $x=a_i$ and $y=b_j$ but $i \neq j$ then the shortest path goes through the cheaper of the two ways (left or right). Specifically, the left path $(a_i,\ell,\ell',b_j)$ has weight $(i-j+n+2)\cdot M + \alpha$ for some $\alpha \{0,1,2\}$ (that depends on the strings: $\alpha=A[i]+B[j]$), and the right $(a_i,r,r',b_j)$ path has weight $(j-i+n+2)\cdot M + \alpha$. Thus, if $i<j$ we choose the left path, and if $i>j$ we choose the right path. In either case, $d(x,y) \leq (n+1) \cdot M+2$.

\item If $x=a_i$ and $y=a_j$ then we again have that $i \neq j$ (or else $x=y$) and the shortest path goes through the cheaper of the two ways (left or right). Specifically, the left path $(a_i,\ell,a_j)$ has weight $(i+j)\cdot M + \alpha$ for some $\alpha \{0,1,2\}$, and the right $(a_i,r,a_j)$ path has weight $(2n+2-i-j)\cdot M + \alpha$. Thus, if $i+j<n+1$ we choose the left path, if $i+j> n+1$ we choose the right path, and if $i+j=n+1$ then both options are equally good. In either case, $d(x,y) \leq (n+1) \cdot M+2$.

\item The case that $x=b_i$ and $y=b_j$ is analogous.

\item Now comes the final case of $x=a_i$ and $y=b_i$. These are the special pairs corresponding to the coordinates and their distances are larger than all the other distances in the graph. This happens because the two paths (left or right) have the same weight and are equally ``bad''. This weight is $(n+2) \cdot M+\alpha$ where $\alpha \in \{0,1,2\}$ is equal to $A[i]+B[i]$.
Therefore, if $A,B$ are disjoint, then for all $i \in [n]$ we have $A[i]+B[i] \leq 1$ and so $d(a_i,b_i) \leq (n+2) \cdot M+1$.
Otherwise, if there is an $i \in [n]$ such that $A[i]=B[i]=1$ then $d(a_i,b_i) = (n+2) \cdot M+2$ which will be the furthest pair in the graph. 
Finally, observe that any path from $x$ to $y$ that uses more than three edges cannot be shortest, since its weights will be at least $(n+3) \cdot M$ and $M>3$.
\end{itemize}

\end{proof}

Thus we have constructed a graph $G$ from the strings $(A,B)$ such that diameter of $G$ is at most $(n+2) \cdot M+1$ if and only if $(A,B)$ are disjoint.
To conclude the proof we describe how a CONGEST algorithm for Diameter leads to a two-party communication protocol.
Assume there is such an algorithm for Diameter with a $T(n)$ upper bound on the number of rounds.
To use this algorithm for their two-party protocol, Alice and Bob look at their private inputs and construct the graph $G$ from our reduction. 
Note that all edges in Alice's part are known to Alice and all edges in Bob's part are known to Bob.
The ``common'' edges which have one endpoint in each side are known to both players since they do not depend on the private inputs.
Then, they can start simulating the algorithm.
In each round, each node $x$ sends an $O(\log{n})$-bit message to each one of its neighbors $y$. For the messages sent on ``internal'' edges $(x,y)$, having both endpoints belong to Alice or to Bob, the players can readily simulate the message on their own without any interaction. This is because all information known to $x$ during the CONGEST algorithm will be known to the player who is simulating $x$.
For the two non-internal edges $(\ell,\ell'),(r,r')$ the two players must exchange information in order to continue simulating the nodes.
This can be done by exchanging four messages of length $O(\log{n})$ at each round. 
At the end of the simulation of the algorithm, some node will know the diameter of $G$ and will therefore know whether $(A,B)$ are disjoint.
At the cost of another bit, both players will know the answer. 
The total communication cost is $T(n) \cdot O(\log{n})$.
\end{proof}

\section{Lower Bounds for Closest Pair of Sets and Hierarchical Clustering}
\label{sec:HC}

In this section we prove a lower bound on the time it takes to simulate the output of the Average-Linkage algorithm, perhaps the most popular procedure for Hierarchical Clustering, in planar graphs, thus proving Theorem~\ref{thm:HC}.
We build on the diamond-like gadget from the simple lower bound for diameter. The constructions will combine many copies of these gadgets into one big graph that is also diamond-like.

\subsection{Preliminaries for the reductions}

The starting point for the reductions in this section is the Orthogonal Vectors problem, which is known to be hard under SETH \cite{Wil04} and the Weighted Clique conjecture \cite{ABDN18}.

\begin{definition}[Orthogonal Vectors]
Given a set of binary vectors, decide if there are two that are orthogonal, i.e. disjoint.
\end{definition}

We consider two variants of the closest pair problem. 

\begin{definition}[Closest Pair of Sets with Max-distance]
Given a graph $G=(V,E)$, a parameter $\Delta$, and disjoint subsets of the nodes $S_1,\ldots,S_{m} \subseteq V$, decide if there is a pair of sets $S_i,S_j$ such that
$$
\MaxD (S_i,S_j) = \max_{u \in S_i, v \in S_j} d(u,v) \leq \Delta.
$$
\end{definition}

In the second variant we look at the sum of all pairs within two sets, rather than just the max. 
This definition is used in the Average-Linkage heuristic and it is important for its success.

\begin{definition}[Closest Pair of Sets with Sum-distance]
Given a graph $G=(V,E)$, a parameter $\Delta$, and disjoint subsets of the nodes $S_1,\ldots,S_{m} \subseteq V$, decide if there is a pair of sets $S_i,S_j$ such that
$$
\SumD (S_i,S_j) = \Sigma_{u \in S_i, v \in S_j} d(u,v) \leq \Delta.
$$
\end{definition}

We could also look at the Min-distance. However, it is easy to observe that the corresponding closest pair of sets problem is solvable in near-linear time. It is enough to sort all the edges and scan them once until a non-internal edge is found. Interestingly, there is also a popular heuristic for hierarchical clustering based on Min-distance, called Single Linkage, and it known that Single-Linkage can be computed in near-linear time in planar graphs.

\subsection{Reduction with Max-distance and Complete Linkage}

We start with a simpler reduction which works only in the Max-distance case. The reduction to Sum-distance will be similar in structure but more details will be required.

\begin{theorem}
\label{thm:MaxD}
Orthogonal Vectors on $n$ vectors in $d$ dimensions can be reduced to Closest Pair of Sets with Max-distance in a planar graph on $O(nd)$ nodes with edge weights in $[O(d)]$. 
The graph can be made unweighted by increasing the number of nodes to $O(nd^2)$.
\end{theorem}

\begin{proof}

Let $v_1,\ldots,v_n \in \{0,1\}^d$ be an input instance for Orthogonal Vectors and we will show how to construct a planar graph $G$ and certain subsets of its nodes from it.
For each vector $v_k, k\in [n]$ we have a set of $2d$ nodes $S_k=\{ u_{k,1},\ldots, u_{k,d} \} \cup \{ u'_{k,1},\ldots, u'_{k,d} \}$ in $G$. 
Each coordinate $v_k[j]$ is represented by two nodes $u_{k,j}$ and $u'_{k,j}$.
In addition, there are two extra nodes in $G$ that we denote $\ell$ and $r$.
Thus, $G$ contains the $2nd+2$ nodes $S_1 \cup \cdots \cup S_n \cup \{\ell,r \}$.
The edges of $G$ are defined in a diamond-like way as follows.
Every node $u_{k,j}$ or $u'_{k,j}$ is connected with a left-edge to $\ell$ and with a right-edge to $r$.
Thus, $G$ is planar.

The crux of the construction is defining the weights, and it will be done in the spirit of our gadget from the diameter lower bound.
Set $M=4$ as before, and for each $k \in [n]$ and $j \in [d]$ we define:
 \begin{align*}
 w(u_{k,j}, \ell) & =   \begin{cases}
      j \cdot M, & \text{if}\ v_k[j]=0 \\
      j \cdot M + 1, & \text{if}\ v_k[j]=1 
    \end{cases}
        \\
 w(u'_{k.j}, \ell) & =   \begin{cases}
      (2d+1-j) \cdot M, & \text{if}\ v_k[j]=0 \\
      (2d+1-j) \cdot M + 1, & \text{if}\ v_k[j]=1 
    \end{cases}
    \\
 w(u_{k.j}, r) & =   \begin{cases}
      (2d+1-j) \cdot M, & \text{if}\ v_k[j]=0 \\
      (2d+1-j) \cdot M + 1, & \text{if}\ v_k[j]=1 
    \end{cases}
    \\
     w(u'_{k.j}, r) & =   \begin{cases}
j \cdot M, & \text{if}\ v_k[j]=0 \\
      j \cdot M + 1, & \text{if}\ v_k[j]=1 
          \end{cases}
 \end{align*}

Note that all weights are positive integers up to $O(\log{n})$. 

\begin{claim}
For any two sets $S_a,S_b$ we have that 
$$
 \MaxD(S_a,S_b) =   \begin{cases}
      \leq (2d+1) \cdot M+1, & \text{if $v_a,v_b$ are orthogonal} \\
      (2d+1) \cdot M + 2, & \text{otherwise.}  
    \end{cases}
$$
\end{claim}

\begin{proof}
The proof is similar to Claim~\ref{cl:gadget_main} since the subgraph of $G$ induced by two sets $S_a,S_b$ (and the shortest paths between them) is similar to our construction for the diameter lower bound (with $2d$ nodes instead of $n$).

Let $x \in S_a, y \in S_b$ be a pair of nodes, and note that the shortest path between them has only two options: it can either go left (via $\ell$) or right (via $r$).
This is because $M>3$ and any other path will have to use more than two edges which means that it has a subpath of the form $\{r,x,\ell\}$, which has cost at least $(2d+1)\cdot M$ for any $x$, which makes the total weight at least $(2d+2) \cdot M$, but there is always a two-edge path with weight at most $(2d+1) \cdot M + 2$.

 We  divide the analysis to three possible cases:
\begin{itemize}
\item If $x_i=u_{a,i}$ and $y_j=u_{b,j}$ for some $i,j \in [d]$, then their distance is exactly $\min\{ (i + j), (4d+2 - i - j)  \} \cdot M + \alpha$ where $\alpha = v_a[i]+v_b[j]$. From the first term (the left path), the distance is at most $2d \cdot M + 2$.
\item If $x_i=u'_{a,i}$ and $y_j=u'_{b,j}$ for some $i,j \in [d]$, then their distance is also $\min\{ (4d+2 - i - j) , (i + j) \}\cdot M + \alpha$ where $\alpha = v_a[i]+v_b[j]$. Now from the second term (the right path), the distance is at most $2d \cdot M + 2$.
\item Finally, the more interesting case is when $x_i=u_{a,i}$ and $y_j=u'_{b,j}$ for some $i,j \in [d]$ (or vice versa, w.l.o.g.), then their distance is $\left( 2d+1 + \min\{ (i-j) , (j-i) \} \right)\cdot M + \alpha$ where $\alpha = v_a[i]+v_b[j]$. Therefore, if $i \neq j$ the distance is again at most $2d \cdot M + 2$. The only case in which the distance is larger, is when $i = j$, in which case we get $(2d+1) \cdot M + v_a[i]+v_b[i]$. 
\end{itemize}
Therefore, $\MaxD(S_a,S_b) = \max_{i \in [d]} (2d+1) \cdot M + v_a[i]+v_b[i]$ and the claim follows.

\end{proof}

Thus, solving the closest pair problem on $G$ with $\Delta = (2d+1)\cdot M +1$ gives us the solution to Orthogonal Vectors.

The reduction can be made to produce an unweighted graph by subdividing each edge of weight $w$ into a path of length $w$. 
The created nodes do not belong to any of the sets.
The total number of nodes is $O(nd^2)$.

\end{proof}

Next, we present an argument based on this reduction showing that the Complete-Linkage algorithm for hierarchical clustering cannot be sped up even if the data is embedded in a planar graph. 
We give a reduction only to the weighted case; the unweighted case remains open (and seems doable but challenging). 

\begin{theorem}
\label{thm:CL}
If for some $\eps>0$ the Complete-Linkage algorithm on $n$ node planar graphs with edge weights in $[O(\log{n})]$ can be simulated in $O(n^{2-\eps})$ time, then SETH is false.
\end{theorem}

\begin{proof}
To refute SETH it is enough to solve OV on $n$ vectors of  $d= O(\log{n})$ dimensions in $O(n^{2-\eps})$ time, for some $\eps>0$.
Given such an instance of OV, we construct a planar graph $G$ such that the solution to the OV instance can be inferred from a simulation of the Complete-Linkage algorithm on $G$.

The graph $G$ is similar to the one produced in the reduction of Theorem~\ref{thm:MaxD} with a few additions described next.
First, we add $M'=11d$ to all the edge weights in $G$. This does not change any of the shortest paths, because for all pairs $s,t$ the shortest path has length exactly one if they are adjacent and exactly two otherwise. 
Then, we connect the nodes of each set $S_i$ with a path such that $u_{i,j}$ is connected to $u_{i,j+1}$ for all $j\in[d-1]$, $u_{i,d}$ is connected to $u_{i,1}'$, and $u_{i,j}'$ is connected to $u_{i,j+1}'$ for all $j \in [d-1]$. All these new edges have weight $M+1=5$.
As a result, all nodes within $S_i$ are at distance up to $5 \cdot 2d = 10d$ from each other, but the distance from any $u_{i,j}$ or $u_{i,j}'$ to $\ell$ or $r$ does not decrease (since the new edges are at least as costly as the difference between, e.g., $w(u_{i,j},\ell)$ and $w(u_{i,j+1},\ell)$).

Next, we analyze the clusters generated by an execution of the Complete-Linkage algorithm on $G$: we argue that at some point in the execution, each $S_i$ will be its own cluster (except that the nodes $\ell,r$ will be included in one of these clusters), and that the next pair to be merged is exactly the closest pair of sets (in max-distance).
This is because the algorithm starts with each node in its own cluster, and at each stage, the pair of clusters of minimum Max-distance are merged into a new cluster.
Let the \emph{merge-value} of a stage be the distance of the merged cluster, and observe that this value does not decrease throughout the stages.
The first few merges will involve pairs of adjacent nodes on the new paths we added, in some order (that depends on the tie-breaking rule of the implementation, which we do not make any assumptions about), and the merge value will be $5$. 
After all adjacent pairs are merged, two adjacent clusters will be merged, increasing the merge-value to $10$.
This continues until the merge value gets to $10d$, and at this point, each $S_i$ is its own cluster (since their inner distance is at most $10d$ and their distance to any other node is larger), plus the two clusters $\{\ell\},\{r\}$.
Next, the merge value becomes $M'+2dM$ and each of the latter two clusters will get merged into one of the $S_i$'s (could be any of them).
At this point, the max-distance between any pair of clusters is exactly the max-distance between the corresponding two sets $S_a,S_b$. This is because the nodes $\ell,r$ will not affect the max-distance.
And so if we know the next pair to be merged, we will know the closest pair and can therefore deduce the solution to OV.

\end{proof}

\subsection{Reduction with Sum-distance and Average Linkage}

The issue with extending the previous reductions to the Sum-distance case is that pairs $i,j$ with $i\neq j$ will contribute to the score (even though their distance is designed to be smaller than that of the pairs with $i=j$). Indeed, if we look at $\SumD(S_a,S_b)$ instead of $\MaxD(S_a,S_b)$ for two vectors $a,b$ we will just get some fixed value that depends on $d$ plus $|a|+|b|$ (the hamming weight of the two vectors, i.e. the number of ones). Finding a pair of vectors with minimum number of ones is a trivial problem, since the objective function does not depend on any interaction between the pair. To overcome this issue, we utilize a degree of freedom in our diamond-like gadget that we have not used yet: so far, the left and right edges both have a $+v[i]$ term, but now we will gain extra hardness by choosing two distinct values for the two edges. 
The key property of the special pairs $i,j, i=j$ that we will utilize is not that their distance is larger, but that their left and right paths are equally long. Thus the shortest path can choose either path depending on the lower order terms of the weights, whereas for the non-special pairs the shortest path is constrained by the high order terms. 

The starting point for the reduction will be the Closest Pair problem on binary vectors with hamming weight.
Alman and Williams \cite{AW15} gave a reduction from OV to the bichromatic version of this problem, and very recently a surprising result of C.S. and Manurangasi \cite{SM19} showed that the monochromatic version (which is often easier to use in reductions, as we will do) is also hard. 

\begin{definition}[Hamming Closest Pair]
Given a set of binary vectors, output the minimum hamming distance between a pair of them.
\end{definition}

\begin{theorem}[\cite{SM19}]
Assuming OVH, for every $\eps>0$, there exists $s_\eps>0$ such that no algorithm running in time $O(n^{2-\eps})$ can solve Hamming Closest Pair on $n$ binary vectors in $d=(\log{n})^{s_\eps}$ dimensions.
\end{theorem}

Next we adapt the reduction from Theorem~\ref{thm:MaxD} to the sum-distance case.

\begin{theorem}
\label{thm:sumD}
Hamming Closest Pair on $n$ vectors in $d$ dimensions can be reduced to Closest Pair of Sets with Sum-distance in a planar graph on $O(nd)$ nodes with edge weights in $[O(d)]$. 
The graph can be made unweighted by increasing the number of nodes to $O(nd^2)$.
\end{theorem}

\begin{proof}

The construction of the planar graph $G$ from the set of vectors will be similar, with one modification in the weights, to the one in Theorem~\ref{thm:MaxD} but the analysis will be quite different.

As before, for each vector $v_k, k \in [n]$ we have a set of $2d$ nodes $S_k=\{ u_{k,1},\ldots, u_{k,d} \} \cup \{ u'_{k,1},\ldots, u'_{k,d} \}$ in $G$, and we have two additional nodes $\ell,r$.
Each node $u_{k,j}$ or $u_{k,j}'$ is connected to both $\ell$ and $r$.

Set $M=4$ as before and for each $i \in [n],j \in [d]$ we define the edge weights of $G$ as follows.
The difference to the previous reduction is that in the edges to $r$ we add the complement of $v_k[j]$ rather than  $v_k[j]$ itself.

\begin{claim}
For any two vectors $a,b$:
$$
\SumD(S_a,S_b) = f(d,M) +2 \cdot \HamD(v_a,v_b)
$$
where $f(d,M)= O(Md^3)$ depends only on $d$ and $M$.
\end{claim}

\begin{proof}
From the above analysis we get:
$$
\SumD(S_a,S_b) = \sum_{i,j \in [d]}  \left( d(u_{a,i} , u_{b,j} ) + d(u_{a,i}', u_{b,j}')  + d(u_{a,i} , u_{b,j}' ) + d(u_{a,i}' , u_{b,j} )  \right)
$$
$$
=    d  (|a| + |b|) + d  (d-|a| + d-|b|) + 2 \cdot \HamD(v_a,v_b)+  2 f_1(d,M) + f_2(d,M) 
$$
which is equal to the claimed expression with $f(d,M) = 2 f_1(d,M) + f_2(d,M) + 2d^2$.
\end{proof}

Thus, the closest pair of sets $S_a,S_b$ in $G$ will correspond to the pair of vectors $a,b$ that minimize $\HamD(v_a,v_b)$. This completes the reduction.
As before, the graph can be made unweighted by subdividing the edges into paths.
\end{proof}

Finally, we present a lower bound argument for the Average-Linkage algorithm in planar graphs. 
As before, the unweighted case remains open. 

\begin{theorem}
\label{thm:AL}
If for some $\eps>0$ the Average-Linkage algorithm on $n$ node planar graphs with edge weights in $[O(\log{n})]$ can be simulated in $O(n^{2-\eps})$ time, then SETH is false.
\end{theorem}
\begin{proof}
The proof is similar in structure to the proof of Theorem~\ref{thm:CL}.
Our graph $G$ will be produced from the graph in the reduction of Theorem~\ref{thm:sumD} by making the following changes:
First, we add $M'=11d$ to all the edge weights. This does not change any of the shortest paths, because for all pairs $s,t$ the shortest path has length exactly one if they are adjacent and exactly two otherwise. 
Then, we connect the nodes of each set $S_i$ with a path such that $u_{i,j}$ is connected to $u_{i,j+1}$ for all $j\in[d-1]$, $u_{i,d}$ is connected to $u_{i,1}'$, and $u_{i,j}'$ is connected to $u_{i,j+1}'$ for all $j \in [d-1]$. All these new edges have weight $M+1=5$.
This makes it so that all nodes within $S_i$ are at distance up to $5 \cdot 2d = 10d$ from each other, but the distance from any $u_{i,j}$ or $u_{i,j}'$ to $\ell$ or $r$ does not decrease.
Finally, we increase the \emph{weight} of all \emph{nodes} in the $S_i$ sets, diminishing the influence that the $\ell,r$ nodes might have on the average distance between two clusters.
This can be done, e.g. by replacing each node by $k$ copies that are all connected with edges of weight $\eps$ in a path (as a subpath of the aforementioned path), and connecting each copy to $\ell,r$ in the same way.

Let us analyze the clusters generated by an execution of the Average-Linkage algorithm on $G$: we argue that at some point in the execution, each $S_i$ will be its own cluster (except that the nodes $\ell,r$ will be included in one of these clusters), and that the next pair to be merged is exactly the closest pair of sets (in sum-distance).
Let the \emph{merge-value} of a stage be the (average-)distance of the merged cluster, and observe that this value does not decrease throughout the stages.
The first few merges will involve pairs of adjacent nodes on the new paths we added, in some order (that depends on the tie-breaking rule of the implementation, which we do not make any assumptions about), and the merge value will be $\eps$. 
When the merge value gets to $10d$, each $S_i$ is its own cluster (since their inner distance is at most $10d$ and their distance to any other node is larger), plus the two clusters $\{\ell\},\{r\}$.
Next, the merge value becomes a bit larger and each of the latter two clusters will get merged into one of the $S_i$'s (could by any of them).
At this point, the closest pair of clusters in average-distance allows us to infer that the corresponding two sets $S_a,S_b$ are the closest pair of sets in sum-distance (and the pair $a,b$ that minimize $\HamD(v_a,v_b)$). 
To see this, first notice that all clusters contain exactly $2dk$ nodes, unless they also contain $\ell$ or $r$ or both (in this case we call them \emph{special clusters}). 
From the proof of Theorem~\ref{thm:sumD} we can conclude that $\SumD(S_a,S_b)$ in our modified $G$ is equal to $f'(d,M) +2k^2 \cdot \HamD(v_a,v_b)$ for any sets $S_a,S_b$.
This is because each coordinate with a mismatch now contributes $+k^2$.
Therefore, the average distance between the clusters is $f'(d,M)/(2kd)^2 +  \HamD(v_a,v_b)/2d^2$ unless they are special.
The average distance between special clusters is a bit smaller, and it can be lower bounded by:
$$
\frac{(2kd)^2}{(2kd+2)^2} \cdot \left( f'(d,M)/(2kd)^2 +  \HamD(v_a,v_b)/2d^2 \right) + \frac{2\cdot (2kd)}{(2kd+2)^2} \cdot 11d 
$$
$$
\geq
( 1 -\frac{4kd +4 }{(2kd+2)^2} )\left( f'(d,M)/(2kd)^2 +  \HamD(v_a,v_b)/2d^2 \right) 
$$
If we set $k >> d^2$ the negative terms (from $\ell,r$) become negligible compared to a $\pm 1$ in $\HamD(v_a,v_b)$. 
Therefore, the next cluster we merge must correspond to the pair $a,b$ that minimize $\HamD(v_a,v_b)$, and we can deduce the solution to the closest pair problem.
\end{proof}

\section{Algorithms for Sparsest Cut and Minimum Quotient Cut}
\label{sec:Algs}

In this section we present our algorithms for Sparsest Cut and Minimum Quotient Cut.

\subsection{Proof of Theorem~\ref{thm:sparsestcut:approx}: An $O(1)$-Approximation for Minimum Quotient Cut in
  near-linear time}


We will describe the algorithm in the dual graph, where cuts are cycles.  Thus the input is a
connected undirected planar graph $G$ with positive integral
edge-costs $\cost(e)$ and integral face-weights $\weight(f)$.  Unless
otherwise specified, $n$ denotes the size of $G$.  We denote the sum
of (finite) costs by $P$ and we denote the sum of weights by $W$.  Given
a cycle $C$, the total cost of the edges of $C$ is denoted $\cost(C)$,
and the total weight enclosed by $C$ is denoted $\weight(C)$, while the total weight outside $C$ is denoted by $\weightbar(C)$.  We
denote by $\lambda(C)$ the ratio
$\cost(C)/{\min}\set{\weight(C), \weightbar(C)}$.  The goal is to find a
cycle $C$ that minimizes $\lambda(C)$.  We give a constant-factor
approximation algorithm. The approximation ratio and the running time
depend on a parameter $\epsilon$, which we assume is a constant.

\paragraph{Overview of the Algorithm.}
Assume that the optimal cut is achieved with the cycle $C$.
Our algorithm has two main parts, both of which combine previously
known techniques with a novel idea. Roughly speaking, the goal of the
first part is to find a node $s$ that is close to $C$, i.e. there is a
path of small cost from $s$ to some node in $C$. the second part will
find an approximately optimal cycle $\hat C$ by starting from a
reasonable candidate that can be computed in near-linear time and then
iteratively improving it using the node $s$.
This idea of finding a nearby node (rather than insisting on a node
that is on the optimal cycle, which incurs an extra $O(n)$ factor)
and then using it to fix a candidate cycle is the crucial one that
lets us improve the running time of the quadratic-time
$3.5$-approximation of Rao~\cite{Rao92} by sacrificing somewhat in the
quality of the solution.

The first part uses a recursive
decomposition of the graph with shortest-path cycle separators, in
order to divide the graph into subgraphs such that the total size of
all subgraphs is $O(n\log n)$ and that we are guaranteed that $C$ will be
in one of them, and, moreover, that for each subgraph there are
$O(1/\epsilon)$ candidate portals $s$ such that one of them is
guaranteed to be close to $C$ (if it is there).

 In the second part,
we make use of the construction of Park and Phillips~\cite{ParkPhillips}
that uses a spanning tree to define a directed graph with edge weights
chosen so that the sum of weights of any
cycle of the tree (if all edges have the same direction) is exactly
the total weight of faces enclosed by the cycle.  Using a classical
technique~\cite{Megiddo}, the problem of finding a cycle with small
cost-to-weight ratio is reduced to the problem of finding a
negative-cost cycle.  The latter problem can be solved in planar graphs in
nearly linear time.

From here, the algorithm and 
analysis follow those of Rao's algorithm~\cite{Rao92}.
The quotient of a cycle $\hat C$ is defined to be the cost of $\hat C$
divided by whichever is smaller, the weight enclosed by $\hat C$ or the weight not
enclosed.  However, the negative-cost cycle technique considers only
the weight enclosed.  Rao provides techniques to address this using
weight-reduction steps.  His algorithm assumes it has correctly
guessed a vertex on the cycle, but the techniques can be adapted to 
work when the vertex is merely close to the cycle.

\subsubsection{Outermost loop}

The outermost loop of the algorithm is a binary search for the
(approximately) smallest value $\lambda$ such that there is a cycle
$C$ for which $\lambda(C) \leq \lambda$.  The body of this loop is a
procedure $\find_\epsilon(\lambda)$ that for a given value of $\lambda$ either
(1) finds a cycle $C$ such that
$\lambda(C) \leq (1+\epsilon)^3 4.5 \lambda$ or (2) determines that
there is no cycle $C$ such that $\lambda(C)\leq \lambda$.  The binary
search seeks to determine the smallest $\lambda$ (to within a factor of $1+\epsilon$) for which $\find_\epsilon(\lambda)$ returns a cycle.
For any fixed value of $\epsilon$, 
because the optimal value (if finite) is between $1/W$ and $P$, the
number of iterations of binary search is $O(\log WP)$.

\subsubsection{Cost loop}

The loop of the $\find_\epsilon(\lambda)$ procedure is a search for the
(approximately) smallest number $\tau$ such that there is a cycle $C$
of cost at most $2\tau$ with $\lambda(C)$ not much more than
$\lambda$.  The body of this loop is a procedure
$\find_\epsilon(\lambda, \tau)$ that either (1) finds a cycle $C$ such that
$\lambda(C) \leq (1+\epsilon)^2 4.59 \lambda$ (in which case we say the procedure
\emph{succeeds}) or (2) determines that there is no cycle $C$ such
that $\lambda(C)\leq \lambda$ and $\cost(C) \leq 2\tau$.  The outer
loop tries $\tau=1, \tau=1+\epsilon, \tau=(1+\epsilon)^2$ and so on,
until $\find_\epsilon(\lambda, \tau)$ succeeds.  The number of iterations is
$O(\log P)$ where $\epsilon$ is a constant to be determined.  In
proving the correctness of $\find_\epsilon(\lambda, \tau)$, we can assume that
calls corresponding to smaller values of $\tau$ have failed.

\subsubsection{Recursive decomposition using shortest-path separators}

The procedure $\find_\epsilon(\lambda,\tau)$ first finds a shortest-path tree
(with respect to edge-costs) rooted at an arbitrary vertex $r$.  The
procedure then finds a recursive decomposition of $G$ using balanced
cycle separators with respect to that tree.  Each separator is a
non-self-crossing (but not necessarily simple) cycle $S=P_1 P_2 P_3$,
where $P_1$ and $P_2$ are shortest paths in the shortest-path tree,
and every edge $e$ not enclosed by $S$ but adjacent to $S$ is
adjacent to $P_1$ or $P_2$.  This property ensures that any
cycle that is partially but not fully enclosed by $S$
intersects $P_1$ or $P_2$.  

The recursive decomposition is a binary tree.  Each node of the tree
corresponds to a subgraph $H$ of $G$, and each internal node is labeled
with a cycle separator $S$ of that subgraph.  The children of a node
corresponding to $H$ and labeled $S$ correspond to the subgraph $H_1$
consisting of the interior of $S$ and the subgraph $H_2$ consisting of
the exterior.  (Each subgraph includes the cycle $S$ itself.)  In
$H_1$ and $H_2$, the cycle $S$ is the boundary of a new face, which is called a
\emph{scar}. 
The scar is assigned a weight equal to the sum of the
weights of the faces it replaced.  Each leaf of the binary tree
corresponds to a subgraph with at most a constant number of faces.
We refer to the subgraphs corresponding to nodes as
\emph{clusters}.

One modification: for the purpose of efficiency, each vertex $v$ on the
cycle $S$ that has degree exactly two after scar formation is
\emph{spliced out}: the two edges $e_1,e_2$ incident to $v$ are replaced with a
single edge whose cost is the sum of the costs of $e_1$ and $e_2$.
Clearly there is a correspondence between cycles before splicing out
and cycles after splicing out, and costs are preserved.  For the sake
of simplicity of presentation, we identify each post-splicing-out
cycle with the corresponding pre-splicing-out cycle.

Consecutive iterations of separator-finding alternate balancing number
of faces with balancing number of scars.  As a consequence, the depth
of recursion is bounded by $O(\log n)$ and each cluster has at most
six scars.  (This is a standard technique.)  Because of the
splicing out, the sum of the sizes of graphs at each level of
recursion is $O(n)$.  Therefore the sum of sizes of all clusters is
$O(n \log n)$.

Let $H$ be a cluster.  Because $H$ has at most six scars, there are at
most twelve paths in the shortest-path tree such that any cycle in the
original graph that is only partially in the cluster must intersect at
least one of these paths (these are the two paths $P_1,P_2$ from above).  We call this the \emph{intersection
  property}, and we refer to these paths as the \emph{intersection}
paths.

Because each scar is assigned weight equal to the sum of the weights
of the faces it replaced, for any cluster $H$ and any simple cycle $C$
within $H$, the cost-to-weight ratio for $C$ in $H$ is the same as the
ratio for $C$ in the original graph $G$.

\subsubsection{Decompositions into annuli}

The procedure also finds $1/\epsilon$ decompositions into \emph{annuli},
based on the distance from $r$.  The \emph{annulus} $A[a,b)$ consists of
every vertex whose distance from $r$ lies in the interval $[a,b)$.
The \emph{width} of the annulus is $b-a$.  Let
$\delta = \epsilon\tau$ and let $\sigma =(1+2\epsilon)\tau$.  For
each integer $i$ in the interval $[0, 1/\epsilon]$, the decomposition
$D_i$ consists of the annuli
$A[i \delta,i\delta+\sigma), A[i \delta+\sigma, i\delta+2\sigma), A[i
\delta+2\sigma, i\delta+3\sigma)$ and so on.  Thus the decomposition
$D_i$ consists of disjoint annuli of width $\sigma$.

\subsubsection{Using the decompositions} \label{sec:using-the-decompositions}

The procedure $\find_\epsilon(\lambda,\tau)$ is as follows:
\begin{tabbing}
  search for a solution in each leaf cluster\\
  for each integer $i\in [0, 1/\epsilon]$\\
  \quad \= for each annulus $A[a,b)$ in $D_i$\\
\> \quad \=  for each non-root cluster $Q$\\
\>  \> \quad \= for each $P$ that is the intersection of the annulus with one of the twelve intersection paths of $Q$\\
\>  \> \> \quad \= form an $\epsilon \tau$-net $S$ of $P$ (take nodes that are $\epsilon \tau$ apart)\\
\>  \> \> \> for each vertex $s$ of $S$\\
\>  \>\> \> \quad \= call subprocedure $\rootedfind_\epsilon(\lambda, \tau, s, R)$\\
\> \>\>\>\>where $R$ = intersection of $A[a,b)$ with the parent of cluster $Q$ 
\end{tabbing}
Here $\rootedfind_\epsilon(\lambda, \tau, s, R)$ is a procedure such that if
there is a cycle $C$ in $R$ with the properties listed below then
the procedure finds a cycle $C$ such
that $\lambda(C) \leq 4.5(1+\epsilon^2) \lambda$ (in which case we say
that the call \emph{succeeds}).

The properties are:
\begin{enumerate}
  \item $\lambda(C) \leq \lambda$, and
  \item $(1+\epsilon)^{-1}2\tau < \cost(C) \leq 2\tau$, and
    \item $C$ contains a vertex $v$ such that the minimum cost of a
      $v$-to-$s$ path is at most $\epsilon \tau$.
\end{enumerate}

In the last step of $\find_\epsilon$, the procedure takes the intersection of an annulus
with a cluster.  Let us elaborate on how this is done.  Taking the intersection with an
annulus involves deleting vertices outside the annulus.  Deleting a
vertex involves deleting its incident edges, which leads to faces
merging; when two faces merge, the weight of the resulting face is
defined to be the sum of weights of the two faces.  This ensures that
the cost-to-weight ratio of a cycle is the same in the subgraph as it
is in the original graph.

We show that $\find_\epsilon(\lambda,\tau)$ is correct as follows.  
If the search for a solution in a leaf cluster succeeds or one of the
calls to $\rootedfind_\epsilon$ succeeds, it follows from the construction that
the cycle found meets the criterion for success of $\find_\epsilon$.
Conversely, suppose
that there is a cycle $C$ in $G$ such that $\lambda(C)\leq \lambda$
and $2(1+\epsilon)^{-1}\tau <\cost(C) \leq 2\tau$.  Our goal is to show that
$\find_\epsilon(\lambda,\tau)$ succeeds.  
Let $Q_0$ be the smallest cluster that
contains $C$.  If $Q_0$ is a leaf cluster then the first line ensures
that $\find_\epsilon(\lambda,\tau)$ succeeds.  Otherwise, $Q_0$ has a child
cluster $Q$ such that $C$ is only partially in $Q$.  Therefore by the
intersection property $C$ intersects one of the intersection paths $P$
of $Q$.  Let $v$ be a vertex at which $C$ intersects $P$.  Let $s$ be
the point in the $\epsilon \tau$-net of $P$ closest to $v$.  

Let $d_{\min}$ be the minimum distance from $r$ of a vertex of $C$, and
let $d_{\max}$ be the maximum distance.  Because $\cost(C) \leq 2\tau$,
we have $d_{\max}-d_{\min} \leq \tau$.  Let $d_{\min}' = {\min} \set{d_{\min},
  \text{distance of $s$ from $r$}}$ and let $d_{\max}'={\max} \set{d_{\max},
  \text{distance of $s$ from $r$}}$.  Then $d_{\max}'-d_{\min}' \leq \tau +
\epsilon \tau$, so there exists an integer $i\in [0, 1/\epsilon]$ and
an integer $j\geq 0$ such that the interval $[i\delta+j\sigma,
i\delta+(j+1)\sigma)$ contains both $d_{\min}'$ and $d_{\max}'$, and
therefore the annulus $A[i\delta+j\sigma,
i\delta+(j+1)\sigma)$ contains $C$ together with the $v$-to-$s$
subpath of $P$.  The specification of $\rootedfind_\epsilon(\lambda, \tau, s,
R)$ therefore shows that the procedure succeeds.

Now we consider the run-time analysis.  The sum of sizes of all leaf
clusters is $O(n)$.  Because each leaf cluster has at most a constant
number of faces, therefore, solutions can be sought in each of the
leaf clusters in a total of $O(n)$ time.

For each integer $i\in [0, 1/\epsilon]$, the annuli of decomposition
$D_i$ are disjoint.  Because the sum of sizes of clusters is
$O(n \log n)$, the sum of sizes of intersections of clusters with
annuli of $D_i$ is $O(n \log n)$. 
Moreover, note that the total size of the $\epsilon \tau$-nets we pick within any annulus of width $O(\tau)$ is $O(1/\epsilon)$.
 Therefore
$O(\epsilon^{-1} n \log n)$ is a bound on the sum of sizes of all
intersections $R$ on which $\rootedfind_\epsilon$ is called.  Therefore in order
to obtain a near-linear time bound for $\find_\epsilon$, it suffices to prove a
near-linear time bound for $\rootedfind_\epsilon$.

\subsubsection{$\rootedfind_\epsilon$}

It remains to describe and analyze $\rootedfind_\epsilon(\lambda, \tau, s, R)$.
We use a construction of Park and Phillips~\cite{ParkPhillips}
together with approximation techniques of Rao~\cite{Rao92}.

Let $T$ be a shortest-path tree of $R$, rooted at $s$.  Delete from the graph
every vertex whose distance from $s$ exceeds $(1+\epsilon)\tau$, and
all incident edges, merging faces as before.  This includes deleting
vertices that cannot be reached from $s$ in $R$.  Let $\hat R$ denote
the resulting graph.  Note that a cycle $C$ in $R$ that satisfies
Properties~2 and~3 (see Section~\ref{sec:using-the-decompositions})
must also be in $\hat R$.

According to a basic fact about planar embeddings (see
e.g.~\cite{planarity}), in the planar dual $\hat R^*$ of $\hat R$, the
set of edges not in $T$ form a spanning tree $T^*$.  Each vertex of
$T^*$ corresponds to a face in $\hat R$ and therefore has an
associated weight.  The procedure
arbitrarily roots $T^*$, and finds the leafmost vertex $f_\infty$ such
that the combined weight of all the descendants of $f_\infty$ is
greater than $W/2$.  The procedure then designates $f_\infty$ as the
infinite face of the embedding of $\hat R$.

\begin{lemma} \label{lem:weight-enclosed} For any nontree edge $e$,
  the fundamental cycle of $e$ with respect to $T$ encloses (with
  respect to $f_\infty$) at most weight $W/2$.
\end{lemma}

Park and
Phillips describe a construction, which applies to any spanning tree
of a planar graph with edge-costs and face-weights, and this
construction is used in  $\rootedfind_\epsilon$.  Each undirected edge of $\hat R$
corresponds to two \emph{darts}, one in each direction.  Each dart is
assigned the cost of the corresponding edge.  A dart corresponding to
an edge of $T$ is assigned zero weight.  Let $d$ be a nontree dart.
Define $w_d$ to be the weight enclosed (with respect to $f_\infty$) by the
fundamental cycle of $d$ with respect to $T$.  Define the weight of
$d$, denoted $\weight(d)$, to be $w_d$ if the orientation of $d$ in the the fundamental cycle
is counterclockwise, and $-w_d$ otherwise.  We refer to this graph as
the \emph{weight-transfer graph}.

\begin{lemma}[Park and Phillips]  The sum of weights of darts of a counterclockwise
  cycle $C$ is the amount of weight enclosed by the cycle.
\end{lemma}

We adapt an approximation technique of Rao~\cite{Rao92}.  (His method
differs slightly.)  The procedure selects a collection of candidate
cycles; if any candidate cycle has quotient at most $4.5 \lambda$,
the procedure is considered to have succeeded.  We will show that if
$\hat R$ contains a cycle with properties~1-3 (see
Section~\ref{sec:using-the-decompositions}) then one of the candidate
cycles has quotient at most $4.5\lambda$.
 
Recall that $W$ is the sum of weights.  We say a dart $d$ is
\emph{heavy} if $\weight(d) \geq \beta W$, where we set $\beta=1/9$.  For each heavy dart, the
procedure considers as a candidate the fundamental cycle of $d$.

We next describe the search for a cycle in the weight-transfer graph
minus heavy darts.  
Following a basic technique (see~\cite{Lawler,Megiddo}), we define a modified
cost per dart as $\widehat{\cost}(d) = \cost(d) - \lambda \weight(d)$.
A cycle has negative cost (under this cost assignment) if and only
if its ratio of cost to enclosed weight is less than $\lambda$. Note that the actual quotient of such a cycle may be much larger than $\lambda$ since we must divide by the \emph{min} of the weight inside and the weight outside the cycle. Still, the information we get from such a cycle will be sufficient for getting a cycle that has quotient not much larger than $\lambda$.

The procedure seeks a negative-cost cycle in this graph.  Using the
algorithm of Klein, Mozes, and Weimann~\cite{KleinMozesWeimann}, this
can be done in $O(n \log^2 n)$ time on a planar graph of size $n$.

Suppose the algorithm does find a negative-cost cycle $\hat C$.  If
$\hat C$ encloses at most $\alpha W$ weight, where we will set $\alpha=5/9$ then $\hat C$ is a
candidate cycle. (In this case, the denominator in the actual quotient of $\hat C$ is not much smaller.)

Otherwise, the procedure proceeds as follows. (Here, the denominator is much smaller, so we would like to fix $\hat C$ so that it encloses much less weight, but the cost does not increase by much.)  It first modifies the
cycle to obtain a cycle $C_0$ that encloses the same amount of weight
and that includes the vertex $s$.  This step consists in adding to
$\hat C$ the shortest path from $s$ to $\hat C$ and the reverse of
this shortest path.  Because the shortest path is in $\hat R$, this
increases the cost of the cycle by at most $2(1+\epsilon)\tau$. (The new cycle will be easier to fix.)

The new cycle $C_0$ might not be a simple cycle: it has the form
illustrated in Figure~\ref{fig:Near-simple_cycle}: it is mostly a
simple cycle but contains a path and its reverse, such that $s$ is an
endpoint of the path.  We refer to such a cycle as a \emph{near-simple
  cycle}.

\begin{figure} \centering
  \includegraphics{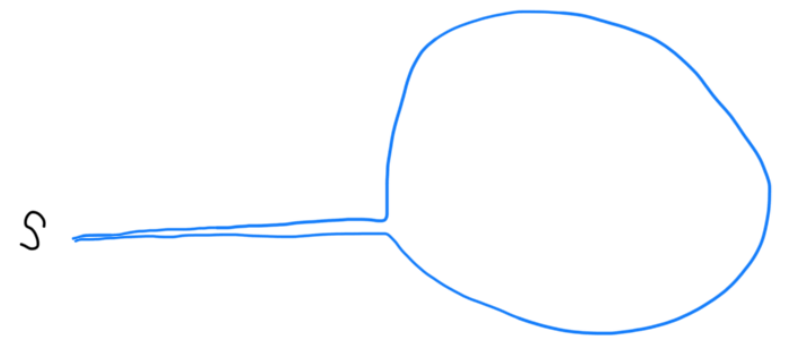}
  \caption{This diagram illustrates the structure of a cycle arising
    in the algorithm.  The cycle is nearly simple but includes a path
    and its reverse, where one endpoint of the path is the root $s$.}
  \label{fig:Near-simple_cycle}
\end{figure}


Next the algorithm iteratively modifies the cycle so as to reduce the
weight enclosed without increasing the cost.  In each iteration, the
algorithm considers the current cycle $C_i$ as a path starting and ending at
$s$, and identifies the last dart $xy$ with $\weight(xy)>0$ in this
path.  The algorithm then finds the closest ancestor $u$ of $x$ in $T$
among vertices occurring after $y$ in the current path.  The
algorithm replaces the $x$-to-$u$ subpath of the current path with the
$x$-to-$u$ path in $T$.  Because the $x$-to-$u$ path in $T$ is a
shortest path, this does not increase the cost of the current path.
It reduces the enclosed weight by at most the weight of $xy$.  This
process repeats until the enclosed weight is at most $\alpha W$.

Here we restate the process, which we call \emph{weight reduction}:
\begin{tabbing}
while $C_i$ encloses weight more than $\alpha W$\\
\quad \= write $C_i = s P_1 xy P_2 s$\\
\> where $xy$ is a positive-weight dart and $P_2$ contains no such dart\\
\> let $u$ be the closest ancestor of $x$ in $T$ among vertices in
$P_2$\\
\> let $C_{i+1} =s P_1 P_3 s$ where $P_3$ is the
$x$-to-$u$ path in $T$\\
\> let $i = i+1$
\end{tabbing}

\begin{lemma} \label{lem:produces-near-simple-cycle}
  The result of each iteration is a near-simple cycle.  The enclosed
  weight is reduced by less than $\beta W$.
\end{lemma}

\subsubsection{Analysis}

We will show that if $\hat R$ contains a cycle $C$ with properties~1-3 (see
Section~\ref{sec:using-the-decompositions}) then one of the candidate
cycles considered by the procedure has quotient at most $4.5\lambda$.  There are three cases.

\subsubsection*{Heavy dart}

Suppose $C$ contains a heavy dart $xy$.  Recall that Property~3 is
that $C$ contains a vertex $v$ such that the minimum cost of a
      $v$-to-$s$ path is at most $\epsilon \tau$.
Let $c_1$ be the cost of the $y$-to-$v$ subpath of $C$, and let $c_2$
be the cost of the $v$-to-$x$ subpath.  Let $u$ be the leafmost common
ancestor of $x$ and $y$ in $T$.  Let $C_{xy}$ denote the fundamental cycle of $xy$
with respect to $T$.  This cycle consists of $xy$, the $y$-to-$u$ path in $T$, and
the $u$-to-$x$ path in $T$.  Let $c_3$ be the cost of the $y$-to-$u$
path in $T$, and let $c_4$ be the cost of the $u$-to-$x$ path.
By the triangle inequality,
\begin{eqnarray*}
  c_3 &\leq& c_1 + \epsilon \tau\\
  c_4 & \leq & c_2 + \epsilon \tau
\end{eqnarray*}               
so $c_3+c_4+\cost(xy) \leq c_1+c_2+2\epsilon \tau+\cost(xy)$, showing
that the cost of $C_{xy}$ exceeds the cost of $C$ by at most
$2\epsilon \tau$.

By Lemma~\ref{lem:weight-enclosed}, $C_{xy}$ encloses weight at most
$W/2$, so its quotient is $\cost(C_{xy})/w_{xy}$.  
Because $xy$ is a heavy dart, the denominator is at least $\beta W$.
Therefore the quotient of $C_{xy}$ is
\begin{eqnarray*}
  \frac{\cost(C_{xy})}{w_{xy}} & \leq & \frac{\cost(C_{xy})}{\beta W}\\
                               & \leq &\frac{\cost(C)+2\epsilon\tau}{\beta W}\\
                               & \leq & \frac{(1+\epsilon(1+\epsilon))\cost(C)}{\beta W}\\
                               & \leq & \frac{1+\epsilon(1+\epsilon)}{2\beta} \frac{\cost(C)}{W/2}\\
                               & \leq &\frac{1+\epsilon(1+\epsilon)}{2\beta} \lambda
\end{eqnarray*}
We chose $\beta=1/9$, so the the quotient is at most $4.5(1+\epsilon +\epsilon^2)\lambda$.

\medskip

\paragraph{No heavy dart.} Suppose that $C$ contains no heavy dart.  In this case, $C$ is a cycle
in the weight-transfer graph minus heavy darts.  The modified cost
function $\widehat{\cost}(\cdot)$ ensures that $C$ is a negative-cost
cycle.  Therefore in this case the algorithm must succeed in finding a
negative-cost cycle $\hat C$.  There are two cases, depending on whether $\hat C$ encloses more or less than $\alpha W$ weight.

\subsubsection*{Small weight inside negative-cost cycle}

Suppose the weight enclosed is at most $\alpha W$.  As a subcase, if
the weight enclosed is at most $W/2$ then the denominator in the
quotient for $\hat C$ is in fact the weight enclosed.  Because $\hat
C$ is a negative-cost cycle, it follows that the quotient of $\hat C$ is less
than $\lambda$.

Suppose therefore that the weight enclosed is greater than $W/2$ but at most
$\alpha W$.  Then the quotient
of $\hat C$ is 
$$\cost(\hat C)/(\text{weight not enclosed by } \hat
C),$$ which is at most
\begin{eqnarray*}
  \frac{\cost(\hat C)}{(1-\alpha)W} & \leq & \frac{\alpha}{1-\alpha}
                                             \frac{\cost(\hat C)}{\alpha W}\\
  & < & \frac{\alpha}{1-\alpha} \lambda
\end{eqnarray*}
The choice $\alpha=5/9$ implies that $\frac{\alpha}{1-\alpha} \leq
4.5$, so in this case the quotient is at most $4.5 \lambda$.

\subsubsection*{Large weight inside negative-cost cycle}

Finally, suppose $\hat C$ encloses more than $\alpha W$ weight. In
this case, the procedure  increases the cost by at most
$2(1+\epsilon)\tau$ and then uses weight reduction.  In each iteration
of weight reduction, the enclosed weight is reduced by an amount that
is less than $\beta W$.  The process stops when the weight is at most
$\alpha W$, so the final weight enclosed is at least $\alpha W - \beta
W$.  Plugging in $\alpha=5/9$ and $\beta=1/9$, we infer that the
denominator in the quotient, the smaller of the weight enclosed and
the weight not enclosed, is at least $(4/9) W$

Therefore the quotient of the final cycle is at most
\begin{eqnarray*}
  \frac{\cost(\hat C)+2(1+\epsilon)\tau}{(4/9) W} & \leq &
              \frac{2(1+\epsilon)^2}{4/9} \frac{\cost(\hat C)}{W}\\
  & \leq & \frac{2(1+\epsilon)^2}{4/9} \frac{\cost(\hat C)}{\text{weight enclosed by } \hat C}\\
 & \leq & \frac{2(1+\epsilon)^2}{4/9} \lambda
\end{eqnarray*}
where the last inequality follows because $\hat C$ is a negative cycle with the modified weights. 

In this case, the quotient is at most $4.5(1+\epsilon^2)]\lambda$.

\subsection{Proof of Theorem~\ref{thm:sparsestcut:exact}: An exact algorithm for Sparsest Cut with running time
  $O(n^{3/2} W)$}
In this section, we provide an exact algorithm for Sparsest Cut and the
Minimum Quotient problems running in time $O(n^{3/2} W \log(C))$. This
improves upon the algorithm 
of Park and Phillips~\cite{ParkPhillips} running in time $O(n^2 W \log(C))$.
We first need to recall their approach (see~\cite{ParkPhillips} for all details).

The approach of Park and Phillips~\cite{ParkPhillips} works as follows.
It works in the dual of the input
graph and thus looks for a cycle $C$ minimizing $\ell(C)/\min(w(\text{Inside}(C)),w(\text{Outside}(C)))$,
where $\ell(C)$ is the sum of length of the dual of the edges of $C$ and $w(\text{Inside}(C))$ (resp.
$w(\text{Outside}(C))$) is the total weight of the vertices of $G$ whose corresponding faces in $G^*$
are in $\text{Inside}(C)$ (resp. $\text{Outside}(C)$). Park and Phillips show that the approach
also works for the sparsest cut problem.
Their algorithm is as follows:
\begin{itemize}
\item[] Step 1. Construct an arbitrary spanning tree $T$ and order the vertices with a preorder
  traversal of $T$ which is consistent with the cyclic ordering of edges around each vertex. 
\item[] Step 2. For each edge $(u,v)$ of $G^*$, create two directed edges $e_1 = \langle u,v \rangle$ and
  $e_2 = \langle v,u \rangle$ and assume $u$ is before $v$ in the ordering computed at Step 1.
  Define the length of $e_1$ and $e_2$ to be the length of
  the dual edge of $e$, define the weight of $e_1$ to be the total weight of vertices
  enclosed by the fundamental cycle induced by $e$ (for the edges of $T$ the
  weight is 0) and the weight of $e_2$ to be minus the weight of $e_1$.
\item[] Step 3. Construct a graph $\mathcal{G}$ as follows: for each vertex $v$ of $G^*$,
  for each weight $y \in [W]$, create a vertex $(v,y)$. For each directed edge $\langle u,v \rangle \in G^*$,
  for each $y \in [W]$, create an edge  between vertices $(u,y)$ and $(v, y+w(\langle u,v \rangle))$ where
  $w(\langle u,v \rangle)$ is the weight of the edge $\langle u,v \rangle$ as defined at Step 2. The length
  of the edge created is equal to the length of $\langle u,v \rangle$.
\end{itemize}

Let $\mathcal{P}$ be the set of
all shortest paths $P_{u,y}$ from $(u,0)$ to $(u, y)$ in $\mathcal{G}$ for $y \in [W/2]$, for each vertex $u \in G^*$.
Let $P^*$ be a shortest path of $\mathcal{P}$ that achieves  $\min_{u \in G^*, y \in [W/2]}\ell(P_{u,y})/y$.
Park and Phillips show that $P^*$ corresponds to a minimum quotient cut of $G$.
The running time $O(n^2 W \log C)$ of the algorithm follows from applying a single source shortest path (SSSP) algorithm
for each vertex $(u,0)$ of $\mathcal{G}$. Since $\mathcal{G}$ has $O(nW)$ vertices, these $n$ SSSP computations
can be done in time $O(n^2W \log C)$.

\paragraph{The Improvement.}
We now show how to speed up the above algorithm. Consider taking an $O(\sqrt{n})$-size balanced separator $S$ of $G^*$.
We make the following observation: either $P^*$ intersects $S$, in which case it is only needed to perform a single source
shortest path computation from each vertex $(u,0)$ of $\mathcal{G}$, where $u \in S$, or $P^*$ does not intersect $S$ and
in which case we can simply focus on computing the minimum quotient cut on each side of the separator separately (treating
the other side as a single face of weight equal to the sum of the weights of the faces in the side).

More precisely, our algorithm is as follows:
\begin{itemize}
\item[] Step 1. Compute an $O(\sqrt{n})$-size balance separator $S$ of $G^*$, that separates $G^*$ into two components
  $S_1,S_2$ both having size $|S_1|,|S_2| \le 2n/3$;
\item[] Step 2. Compute $\mathcal{G}$ and perform an SSSP computation from each vertex $(u,0) \in \mathcal{G}$ where
  $u \in S$ and let $P_0^*$ be the shortest path $P_{u,y}$ from $(u,0)$ to $(u,y)$ for $u \in S$, $y \in [W/2]$ that minimizes
  $\ell(P_{u,y})/y$.
\item[] Step 3. For $i \in {1,2}$, create the graph $G_i$ with vertex set $S_i$ and where the face containing $S_{3-i}$ has weight equal
  to the sum of the weights of the faces in $S_{3-i}$.
\item[] Step 4. Returns the minimum quotient cut among $P^*_0,P_1^*,P^*_2$.
\end{itemize}

The correctness follows from our observation: if $P^*$ intersects $S$ then, by~\cite{ParkPhillips}, Step 2 ensures that $P^*_0$ corresponds
to a minimum quotient cut, otherwise $P^*$ is strictly contained within $S_1$ or $S_2$ and in which case the following argument applies.
Assuming $P^*$ lies completely within $S_1$, then the graph $G_1$ where the face containing $S_2$ has weight equal to the sum of the
weight of faces in $S_2$ contains a minimum quotient cut of value at most the minimum quotient cut of $G$ since
the cut places all the vertices of $S_2$ on one side. Hence, an immediate induction shows that the minimum quotient
cut among the cuts induced
by the paths $P^*_0,P^*_1,P^*_2$  is optimal.
The running time follows from a direct application of the master theorem.

\medskip
\paragraph{Acknowledgments}
P. Klein is supported by NSF Grant CCF-1841954, and V. Cohen-Addad supported by
ANR-18-CE40-0004-01.

\bibliographystyle{plainurl}
\bibliography{refs}

\end{document}